\documentclass[journal]{IEEEtran}
\usepackage[textwidth=1.5in,shadow,backgroundcolor=yellow,textsize=tiny]{todonotes}

\newcommand{\revision}[1]{{\black #1}}

\newif\ifjrnl
\jrnltrue

\ifjrnl
  \newcommand{\omitproofdisclaimer}[1]{}
  \newcommand{\fullproof}[1]{{\color{black} #1}} 
  \newcommand{\proofsketch}[1]{}
  \newcommand{\jrnlonly}[1]{{\color{black} #1}} 
  \newcommand{\tr}[1]{{\color{blue} #1} } 
  \newcommand{\nonjrnlonly}[1]{}
\else
  \newcommand{\omitproofdisclaimer}[1]{#1}
  \newcommand{\fullproof}[1]{}
  \newcommand{\proofsketch}[1]{#1}
  \newcommand{\jrnlonly}[1]{}
  \newcommand{\tr}[1]{}
  \newcommand{\nonjrnlonly}[1]{#1}
\fi

\usepackage{soul}

\usepackage{microtype}
\usepackage[textwidth=1.5in,shadow,backgroundcolor=yellow,textsize=tiny]{todonotes}
\usepackage{pgf,tikz, pgfplots, pgfplotstable}
\usepgfplotslibrary{groupplots}
\usepackage{colortbl}
\pgfplotsset{compat=1.15}
\usetikzlibrary{spy,fit,patterns,patterns.meta}
\usepackage{amsmath,amssymb}

\usepackage[compress]{cite}
\usepackage{amsthm}
\usepackage{tabto}
\usepackage[noend]{algpseudocode}
\usepackage{algorithm}
\usepackage{varwidth} 
\usepackage{xspace}
\usepackage{paralist}
\usepackage[hyphens]{url}
\usepackage[hyphenbreaks]{breakurl}
\usepackage{hyperref}
\usepackage{float}
\usepackage{graphicx, pstricks}
\usepackage{pst-node,pst-plot}
\usepackage{tabu, subfiles, color, multirow}
\usepackage[caption=false]{subfig}

\newcommand{\plotLineWidth}{0.8 pt}
\newcommand{\plotWidth}{0.55*\columnwidth}
\newcommand{\plotHeight}{0.5*\columnwidth}
\newcommand{\plotLegendYOffset}{1.45}

\algnewcommand{\AND}{\textbf{and}\xspace}

\newbox\statebox
\newcommand{\myState}[1]{%
    \setbox\statebox=\vbox{#1}%
    \edef\thealgruleheight{\dimexpr \the\ht\statebox+1pt\relax}%
    \edef\thealgruledepth{\dimexpr \the\dp\statebox+1pt\relax}%
    \ifdim\thealgruleheight<.75\baselineskip
        \def\thealgruleheight{\dimexpr .75\baselineskip+1pt\relax}%
    \fi
    \ifdim\thealgruledepth<.25\baselineskip
        \def\thealgruledepth{\dimexpr .25\baselineskip+1pt\relax}%
    \fi
    \State #1%
    \def\thealgruleheight{\dimexpr .75\baselineskip+1pt\relax}%
    \def\thealgruledepth{\dimexpr .25\baselineskip+1pt\relax}%
}

\tikzdeclarepattern{
  name=hatch,
  parameters={\hatchsize,\hatchangle,\hatchlinewidth},
  bounding box={(-.1pt,-.1pt) and (\hatchsize+.1pt,\hatchsize+.1pt)},
  tile size={(\hatchsize,\hatchsize)},
  tile transformation={rotate=\hatchangle},
  defaults={
    hatch size/.store in=\hatchsize,hatch size=5pt,
    hatch angle/.store in=\hatchangle,hatch angle=0,
    hatch linewidth/.store in=\hatchlinewidth,hatch linewidth=.4pt,
  },
  code={
      \draw[line width=\hatchlinewidth] (0,0) -- (\hatchsize,\hatchsize);
  }
}

\makeatletter
\def\blfootnote{\xdef\@thefnmark{}\@footnotetext}
\makeatother

\newtheorem{theorem}{Theorem}
\newtheorem{corollary}[theorem]{Corollary}
\newtheorem{lemma}[theorem]{Lemma}
\newtheorem{proposition}[theorem]{Proposition}

\newcommand{\abs}[1]{\left\vert#1\right\vert}
\newcommand{\set}[1]{\left\{#1\right\}}

\newcommand{\floor}[1]{\lfloor #1 \rfloor}

\DeclareMathOperator*{\argmin}{arg\,min}

\DeclareMathOperator{\alg}{Alg}
\DeclareMathOperator{\ds}{DS}

\newcommand{\dspgm}{$\ds_{\pgm}$}

\newcommand{\SetCaches}{N} 
\newcommand{\NumCaches}{n} 

\newcommand{\SetP}{N_x} 
\newcommand{\numP}{n_x} 
\newcommand{\mS} {S}   
\newcommand{\ind}{I} 
\newcommand{\accsCost} {c}   

\DeclareMathOperator{\fpr}{FP} 
\DeclareMathOperator{\fnr}{FN} 
\newcommand{\fprj}{\fpr_j} 
\newcommand{\fnrj}{\fnr_j} 

\newcommand{\Phit}{h}
\newcommand{\Pmiss}{(1-\Phit)}
\newcommand{\Phitj}{\Phit_j} 
\newcommand{\mr}{\rho} 
\newcommand{\mrj}{\mr_j} 
\newcommand{\mrone}{\pi} 
\newcommand{\mrzero}{\nu} 
\newcommand{\mronej}{\mrone_j} 
\newcommand{\mrzeroj}{\mrzero_j} 
\newcommand{\missp}{M} 
\newcommand{\costhetero}{\phi} 
\newcommand{\costhomo}{\hat{\costhetero}} 
\newcommand{\pif}{PI} 
\newcommand{\tcpif}{\costhomo^{\pif}} 
\newcommand{\Pone}{q} 
\newcommand{\Ponej}{\Pone_j} 
\newcommand{\ecmfna}{HoCS$_{\fna}$} 
\newcommand{\ecmfno}{HoCS$_{\fno}$} 

\DeclareMathOperator{\fna}{FNA}
\DeclareMathOperator{\fno}{FNO}
\DeclareMathOperator{\pgm}{PGM}
\newcommand{\pgmfna}{HeCS$_{\fna}$} 
\newcommand{\pgmfnastar}{HeCS$_{\fna^*}$} 
\newcommand{\pgmfno}{HeCS$_{\fno}$} 
\newcommand{\bpe}{bpe}
\DeclareMathOperator{\mbpe}{\bpe}

\newcommand{\argminCosthomoOne}{r^*_1} 
\newcommand{\argminCosthomoZero}{r^*_0} 
\newcommand{\subsetCaches}{D} 
\newcommand{\cacheprobFna}{CS$_{\fna}$} 
\newcommand{\cacheprobFno}{CS$_{\fno}$} 
\newcommand{\cacheSize}{C}

\date{}
\title{
False Negative Awareness in
Indicator-based Caching Systems
}

\author{
\IEEEauthorblockN
{
Itamar Cohen\IEEEauthorrefmark{1},
Gil Einziger\IEEEauthorrefmark{2}, and
Gabriel Scalosub\IEEEauthorrefmark{3}
}

\IEEEauthorblockA
{
\IEEEauthorblockA{\IEEEauthorrefmark{1}Department of Electronics and Telecommunications, Politecnico di Torino, Italy}\\
\IEEEauthorblockA{\IEEEauthorrefmark{2}Department of Computer Science, Ben-Gurion University of the Negev, Beer Sheva, Israel}\\
\IEEEauthorblockA{\IEEEauthorrefmark{3}School of Electrical and Computer Engineering, Ben-Gurion University of the Negev, Beer Sheva, Israel}\\
Email: itamar.cohen@polito.it, gilein@bgu.ac.il, sgabriel@bgu.ac.il
}
}

\begin{document}
\maketitle
\begin{abstract}
Distributed caching systems such as content distribution networks often advertise their content via lightweight approximate indicators (e.g., Bloom filters) to efficiently inform clients where each datum is likely cached.
While false-positive indications are necessary and well understood, most existing works assume no false-negative indications.
Our work illustrates practical scenarios where false-negatives are unavoidable and ignoring them significantly impacts system performance.
Specifically, we focus on false-negatives induced by indicator staleness, which arises whenever the system advertises the indicator only periodically, rather than immediately reporting every change in the cache. Such scenarios naturally occur, e.g., in bandwidth-constraint environments or when latency impedes each client's ability to obtain an updated indicator.
Our work introduces novel false-negative aware access policies that continuously estimate the false-negative ratio and sometimes access caches despite negative indications.
We present optimal policies for homogeneous settings and provide approximation guarantees for our algorithms in heterogeneous environments.
We further perform an extensive simulation study with multiple real system traces.
We show that our false-negative aware algorithms incur a significantly lower access cost than existing approaches or match the cost of these approaches while requiring an order of magnitude fewer resources (e.g., caching capacity or bandwidth).
\end{abstract}

\section{Introduction}
\label{Sec:Intro}
\blfootnote{* The work was done while this author was with Ben-Gurion University.}
\blfootnote{An earlier version of this work was published in~\cite{FN_aware_ICDCS}.}
Caches are extensively used in networking environments such as Content Delivery Networks~\cite{summary_cache, CDN_theory_Vs_practice,
Joint_opt}, Named Data Networks~\cite{indicators_in_NDN19}, 5G networks~\cite{5GNetworks},  
and Information-Centric Networks~\cite{ICN_survey_15}.
In such networks, accessing caches often incurs some overhead in terms of latency, bandwidth, or energy~\cite{Joint_opt, BloomParadox}. On the other hand, fetching a datum without caches usually incurs a larger {\em miss penalty}, e.g., for retrieving the requested item from a remote server~\cite{CDN_theory_Vs_practice}. 

In large distributed systems, caches often
optimize performance by advertising their content~\cite{summary_cache, CDN_theory_Vs_practice,Joint_opt, indicators_in_NDN19,ICN_survey_15,
Digest_in_Manet}.
Such advertisements allow clients to minimize costs by selecting which cache to access for a requested datum.
Ideally, the advertisement policy would always accurately reflect the up-to-date content at every cache.
However, such a solution requires a prohibitive amount of memory, computation, and bandwidth resources.
Hence, systems often compromise some accuracy for efficiency by advertising periodical approximate {\em indicators}. Indicators are data structures that trade accuracy for space efficiency.
Typical embodiments of such indicators are Bloom filters~\cite{Bloom, CBF, Survey18, BloomParadox}, and fingerprint hash tables~\cite{TinyTableJournal}.

Such approximations commonly introduce the risk of {\em false-positive} errors, i.e., the indicator sometimes wrongly indicates that a datum is stored in the cache.
In such a case, accessing the cache results in an unnecessary cache access, translating to an excessive cost.
Consequently, the problem of advertising space-efficient indicators while keeping a low false-positive rate has attracted a bulk of research effort~\cite{Bloom, CBF, BloomParadox, Survey12, Survey18}.
Other works addressed the {\em cache selection} problem, namely, selecting which cache to access when there are more than one positive indications, where some of them may actually be false-positives~\cite{BloomParadox, Accs_Strategies_ToN}. 

Most previous works~\cite{BloomParadox,Accs_Strategies_ToN,Joint_opt} assume that there are no {\em false-negative} indications.
Indeed, there exist indicators that theoretically guarantee a false-negative ratio of zero (e.g., a simple Bloom filter~\cite{Bloom}), or a negligible false-negative ratio (e.g., a Counting Bloom Filter (CBF)~\cite{CBF}).
However, to manifest this guarantee in a practical distributed environment, every cache should advertise its indicator to all the clients in the network upon every change in the cached content, usually resulting in prohibitive bandwidth consumption.
For instance, Akamai reports using Bloom filters of about 70MB in size in every cache~\cite{CDN_theory_Vs_practice}.
Insisting on sending an update upon every change in the cached content in such a system may result in the advertised indicators consuming more bandwidth than the cached content itself.
Hence, caches commonly advertise their content only periodically. 

When using periodical updates, the advertised content gradually becomes {\em stale}.
Namely, it takes time for the indicator available at the clients to reflect changes in the cached content.
Unfortunately, such staleness may lead to a significant increase in false-negative indications.
To illustrate the problem, consider a cache that advertises a fresh indicator, and later admits a new item $x$.
When the client tests for $x$, the indicator is likely to wrongly indicate that $x$ is not in the cache (as it wasn't in the cache when the cache sent the advertisement), resulting in a false-negative error.
Such scenarios are quite common in highly dynamic networks, such as 5G networks~\cite{5GNetworks}.

\begin{figure}[th]
    \captionsetup{width=\linewidth}
	\centering
    \subfloat[
    \label{fig:Fn_Vs_uInterval}
    Effect of the update interval on the false-negative ratio. Both axes are in log-scale, the cache size is 10K, and the policy is LRU. The traces are Wiki and Gradle (described in Sec.~\ref{sec:sim}). The distinct plots for increasing values of bits-per-cached-element ($\bpe$), correspond to increasing indicator sizes, with decreasing false-positive ratios, respectively.]{
    \begin{tikzpicture}
        log bases x={2},
        log bases y={10},
        \begin{groupplot}[
            group style=
                {
                columns         = 2,
                xlabels at      = edge bottom,
                ylabels at      = edge left,
                horizontal sep  = 0.07\columnwidth,
                group name      = plots
                },
    		width  = \plotWidth,
     		height = \plotHeight,
    		xmode = log,
    		ymode = log,
    		xmin = 2,
     		xmax = 8192,
     		xtick = {2, 16, 128, 1024, 8192},
     		xticklabels = {2,16,128,1K,8K},
     		ymin = 0.0005,
     		ymax = 0.6,
     		minor tick style={draw=none},
            ymajorgrids = true,
            ylabel = {Ratio},
     		ylabel near ticks,
    		legend style = {at={(-0.1,\plotLegendYOffset)}, anchor=north, legend columns=-1, font=\footnotesize},
    		xlabel = {Update Interval (\# Insertions)}, 
    		xlabel near ticks,
    	    label style={font=\footnotesize},
    	    tick label style={font=\footnotesize},
    	]
        	
        \nextgroupplot[
            title = Wiki,
            title style ={
                font=\footnotesize,
                yshift = -2pt
                },
     		]


\addplot[color=blue,  	mark=x, 				width = \plotwidth] coordinates {
(2, 0.0007408257598681122)
(4, 0.00154193781005413)(8, 0.003211404311165369)(16, 0.005695895083305219)
(32, 0.008935123836155775)(64, 0.014197421369866825)
(128, 0.02304095641880886)(256, 0.0365833440573369)
(512, 0.05779252473187093)(1024, 0.10830779861132848)(2048, 0.17132320291785644)
(4096, 0.2333554769515078)(8192, 0.30226386613661105)
};


        \nextgroupplot[
            title = Gradle,
            title style ={
                font=\footnotesize,
                yshift = -2pt,
                },
      		yticklabels = \empty
      	]

\addplot[color=red, 	 	mark=o, 				width = \plotwidth] coordinates {
(2, 0.0008012383357776986)
(4, 0.0018859588597811517)
(8, 0.003618944312331047)(16, 0.00644306875352347)
(32, 0.011139245381112384)(64, 0.019073537419113974)
(128, 0.03086532670519885)(256, 0.0493098545896836)
(512, 0.07217777435459664)(1024, 0.10739695957330048)
(2048, 0.14588956175793188)(4096, 0.20278390206278757)(8192, 0.263361371326365)
};\addlegendentry {4-$\bpe$}

\addplot[color=blue,  	mark=x, 				width = \plotwidth] coordinates {
(2, 0.0011938344228677059)
(4, 0.0027032433571565345)
(8, 0.004945427004406276)(16, 0.008447769209127485)
(32, 0.014445824414341845)(64, 0.023389954888891028)
(128, 0.03759080790284156)(256, 0.05775120533417201)
(512, 0.09082662336342523)(1024, 0.12636031334944367)
(2048, 0.17487320857977562)(4096, 0.24456061865441167)(8192, 0.319719769833858)
};\addlegendentry {8-$\bpe$}

\addplot[color=purple, mark=+, 				width = \plotwidth] coordinates {
(2, 0.0021138143544682676)(4, 0.0040500511872554965)
(8, 0.00673403915049481)(16, 0.011087897663999795)
(32, 0.018005932800814715)(64, 0.0283696137047057)
(128, 0.04383597399666026)(256, 0.06598074674557099)
(512, 0.10294789383431589)(1024, 0.14181169722390707)(2048, 0.20789406965959673)
(4096, 0.2765769900182177)(8192, 0.3686338191041321)
};
\addlegendentry {16-$\bpe$}

        \end{groupplot} 
    \end{tikzpicture}
    }
    
    \subfloat[
    \label{fig:toy_example}
    The client is looking for item $x$ and needs to select a cache. Accessing each cache incurs some known access cost. Each cache $j$ provides an indication w.r.t. item $x$, where $I(x)=1$ means that the item is likely to be cached, and $I(x)=0$ means that it is likely not cached. The true answer is captured by $C(x)$, which is ``Yes'' iff $x$ is found in the cache.  The indications of caches 1, 2, and 3 for item 3 are true-negative, false-positive, and false-negative, respectively. Failing to retrieve $x$ from a cache incurs a miss penalty of 100.]{
    	\centering
        \small
        \begin{tabular}{|c|c|c|c|}
        \cline{2-4}
        \multicolumn{1}{c|}{} & Cache 1 & Cache 2 & Cache 3 \\
        \hline
        Cost & 10 & 20 & 1 \tabularnewline
        \hline
        Indication & I(x) = 0 & I(x) = 1 & I(x) = 0 \tabularnewline
        \hline
        True answer & C(x) = No & C(x) = No & C(x) = Yes \tabularnewline
        \hline
        \end{tabular}
    }
    \label{fig:motivation}
    \caption{Motivation for a false-negative aware approach.}
\end{figure}

To explore the significance of false-negatives caused due to staleness, consider Fig.~\ref{fig:Fn_Vs_uInterval}.
The figure presents the false-negative ratio indications as a function of the time between subsequent advertisements, referred to as the {\em update interval} (both axes are in logarithmic scale).
We measure the update interval by the number of cache changes (insertions of new items).
The indicator used is an optimally configured simple Bloom filter~\cite{Bloom}, where the figure shows distinct indicators with varying number of Bits Per cached Element (\bpe); a higher \bpe\ implies a larger indicator, that is guaranteed to provide a lower false-positive ratio~\cite{Survey12}.
Fig.~\ref{fig:Fn_Vs_uInterval} shows that the false-negative ratio dramatically increases for all indicator sizes. Furthermore, this phenomenon is manifested for various types of workloads, where Fig.~\ref{fig:Fn_Vs_uInterval} shows this for two specific traces, Wiki and Gradle (which represent significantly distinct workloads, as described in Sec.~\ref{sec:sim_settings}).
For instance, it is not uncommon to have a false-negative ratio as high as 10\% when the update interval is above 1K.
Most interestingly, using a larger indicator, which guarantees a lower inherent false-positive ratio~\cite{Survey12}, results in a {\em higher} false-negative ratio.
We discuss and explain this phenomenon in more detail in Sec.~\ref{sec:sim:ind_parameters}.

Fig.~\ref{fig:toy_example} exemplifies the potential benefit of a false-negative aware access strategy. A false-negative oblivious strategy would access only caches with positive indication, i.e., cache 2. However, this indication is a false-positive, thus incurring a miss penalty of 100 for a total cost of 120. In contrast, a false-negative-aware approach may access both caches 2 and 3, implying an access cost of 21 and a hit. Intuitively, as the miss penalty is 100 while the access cost of cache 3 is only 1, it is beneficial to access cache 3 despite a negative indication if the likelihood of a false-negative event is more than $1\%$. 

Designing an access strategy that considers both false-positives and false-negatives
is a challenging task. In particular, it is unclear whether, or when, it may be beneficial to access a cache despite a negative indication.
Furthermore, it is unclear even how to estimate the probability of a false-negative event.
However, to the best of our knowledge, despite its importance, this problem has never been studied.

We stress that we do not address cache replacement policies, the design of new indicators, or optimizing update intervals.
We focus our attention on a false-negative-aware {\em cache access strategy}, namely, which cache(s) should the client access, being aware of the non-negligible probability of false-negatives, which naturally arises in scenarios where indicators become stale.

\subsection{Our Contribution}

We consider the problem of accessing a multi-cache system while using indicators that exhibit both false-positive and false-negative indications.
We challenge the common practice, which assumes that it is {\em always} better not to access caches with negative indications.
Specifically, we develop a framework that supports {\em false-negative awareness}, and design policies that actively access caches with negative indications, aiming at minimizing the overall access cost.

After presenting our system model and some preliminaries in Sec.~\ref{sec:model}, we consider Bloom filters as our main example of indicators (Sec.~\ref{sec:bloom_filters}).
Specifically, we discuss how staleness affects the false-negative and false-positive ratios manifested in such indicators.

We then turn to present an algorithm for fully-homogeneous environments, and show that it is optimal in terms of the overall access cost,
assuming that the false-negative ratio, the false-positive ratio, and the hit ratio, are all known to the client.
These results appear in Sec.~\ref{sec:homo}.
Later, in Sec.~\ref{sec:hetero}, we develop a strategy for realistic environments that are both heterogeneous (i.e., distinct caches have distinct attributes) and dynamic (i.e., the client may not know the precise attributes of each cache). We show how both the cache and the client may cooperate, allowing the client to estimate some of the underlying distributions (which may depend inter alia on the system configuration and the workload being served).
We make explicit use of our analysis of Bloom filters, to serve as a concrete example of how to implement our approach.

Our suggested false-negative aware ($\fna$) approach makes deliberate accesses also
to caches with negative indications.
Furthermore, we show that any approximation guarantee provided by a false-negative oblivious ($\fno$) access strategy (in our model), can be used by our false-negative aware framework.
In particular, we show how to employ known $\fno$ strategies as subroutines, which induce their performance guarantees on our proposed $\fna$ solution.

Finally, in Sec.~\ref{sec:sim}, we present the results of our in-depth simulation study, where we evaluate the performance of our proposed solution in varying system configurations.
We show that our $\fna$ strategy implies a significant reduction in access costs in many real-life scenarios, compared to state-of-the-art $\fno$ approaches.
Furthermore, our results show that our $\fna$ strategy with minimal resources obtains comparable results to those obtained by $\fno$ strategies that use considerably more resources.
For instance, our results indicate that to match the performance of our $\fna$ strategy, an $\fno$ approach might require as much as an order of magnitude more resources (e.g., in terms of system caching capacity or the bandwidth required for indicator advertisement).

\subsection{Related Work}

Indicators are typically used to periodically advertise the content of caches efficiently. Indicators are used in multiple networking environments, including wide-area networks~\cite{summary_cache}, content delivery networks~\cite{CDN_theory_Vs_practice,Joint_opt}, information-centric networking~\cite{indicators_in_NDN19,ICN_survey_15}, and wireless networks~\cite{uIntervalInMANET, Digest_in_Manet}.

Since indicators are of bounded size, they usually fail to represent the cache content accurately and exhibit false-positive indications~\cite{Bloom, Survey12}. 
The pioneering work of~\cite{BloomParadox} shows that due to these false-positives, sometimes naively relying on an indicator for accessing even a single cache may do worse than not using an indicator at all.
Subsequent works~\cite{Accs_Strategies_ToN, chen2020sequential} tackle a distributed scenario where multiple caches advertise indicators, and develop access strategies that take into account both the access cost, and the false-positive ratio in each cache, to minimize the overall expected cost. However,~\cite{BloomParadox, Accs_Strategies_ToN, chen2020sequential} disregard {\em false-negative} indications. 

The work of~\cite{fnr_in_CBFs} studies the problem of false-negatives in practical deployment of Counting Bloom Filters~\cite{CBF}.
Other techniques to reduce the false-negative ratio in numerous variants of Bloom filters are surveyed  in~\cite{Survey18}.
However, these works address false-negatives that stem from {\em architectural design}, i.e., from concrete data structures used to implement indicators. Consequently, these works focus on developing enhanced data structures that reduce such false-negatives.
In contrast, we focus on false-negatives caused by staleness, i.e., false-negatives that follow from the {\em operational usage} of the system.
Such false-negatives may occur in {\em any} indicator, even if its design is false-negative-free, such as a simple Bloom filter~\cite{Bloom, Survey12}. 
In this sense, our approach is orthogonal to previous work targeting the reduction of false-negatives~\cite{fnr_in_CBFs, Survey18}, and these approaches may be seamlessly combined with our solutions.  

Since constantly advertising a fresh indicator might be prohibitively costly,
in practice caches commonly advertise fresh indicators only periodically~\cite{summary_cache, fpr_fnr_in_dist_replicas, uIntervalInMANET, uIntervalSimsInCCN, CAB}, where one usually refers to the period between the advertisements of fresh indicators as the {\em update interval}.
Several works
address
the interplay between the update interval and performance by means of simulations~\cite{summary_cache,uIntervalInMANET}.
The work~\cite{CAB}
presents
an algorithm that dynamically scales the update interval and the indicator size, to comply with bandwidth constraints.
The works~\cite{FPfree_Ori, HBA_journal} reduce the transmission overheads by accurately advertising important information, while allowing less important information to be stale, or less accurate.  
The work~\cite{fpr_fnr_in_dist_replicas}
analyzes
the impact of stale Bloom filter replicas on the false-positive ratio and the false-negative ratio.
However, the framework of~\cite{fpr_fnr_in_dist_replicas} implicitly assumes that requests are drawn from a uniform distribution, and that each object is stored in a single cache.
This framework conforms primarily with distributed storage systems. 
However, in many (if not most) real-life distributed caching environments, requests need not be drawn from a uniform distribution, and furthermore objects may be found in either a single cache, multiple caches, or no cache at all~\cite{indicators_in_NDN19,ICN_survey_15,Joint_opt}.
Part of our analysis of such general environments is inspired by ideas introduced in~\cite{fpr_fnr_in_dist_replicas} (see, e.g., Sec.~\ref{sec:cache-side-alg}).

The problem of stale indicators relates to other problems of decision-making under uncertainty. In particular, our problem is closely related to the concept of the {\em Age of Information} (AoI).
The AoI quantifies the time since the generation of the last successfully received information from a remote system.
The AoI paradigm was applied to numerous environments, e.g., vehicular networks, scheduling, and buffer management; a detailed survey can be found in~\cite{AgeOfInfo_survey_17}.
The AoI was also applied to caches~\cite{AoI_in_cache_data}, but in the context of the coherency of the cached data, while we focus on the coherency of the indicators. 

\jrnlonly{
Lastly, the question of whether to follow the recommendation of a binary indicator has been extensively studied in the context of branch prediction~\cite{CLBP}.
However, the models used to study such systems significantly differ from those considered in our work.
In particular, branch-prediction models do not adhere to or follow traditional cache-memory models, which lay at the core of our work.
}


\section{System Model and Preliminaries}\label{sec:model}

\begin{table}[t]
	\centering
	\caption{\label{tbl:bf:notations}List of Symbols. The top part corresponds to our system model (Sec.~\ref{sec:model}), the middle part corresponds to the structure of Bloom filters (Sec.~\ref{sec:bloom_filters}), and the bottom part corresponds to the fully-homogeneous case (Sec.~\ref{sec:homo}), and the heterogeneous settings (Sec.~\ref{sec:hetero}-\ref{sec:sim}).}
	 	\begin{tabular}{|m{1cm}|m{7.0cm}|}   
		\hline
		Symbol & Meaning \tabularnewline
		\hline
		$\SetCaches$ & Set of caches.\tabularnewline
		\hline
		\revision{$[i]$}& \revision{The set of integers $\set{0, \ldots, i}$} \tabularnewline
    	\hline
		$\NumCaches$ & Number of caches: $\NumCaches = \abs{\SetCaches}$ \tabularnewline
		\hline		
		 $\SetP$ & Set of caches with positive indications \tabularnewline
		    \hline
		$\numP$ & Number of positive indications: $\numP = \abs{\SetP}$ \tabularnewline
		\hline
		$\mS_j$ & The set of data items in cache $j$\normalsize
		\tabularnewline
		\hline
		$\Phitj$ & Hit ratio of cache $j$: $\Phitj = \Pr (x \in S_j)$ \tabularnewline
		\hline
		$\ind_j$ & Indicator of cache $j$ \tabularnewline
		\hline
		$\ind_j (x)$ & Indication of indicator $\ind_j$ for item $x$ \tabularnewline
		\hline
		$\fprj$ & False-positive ratio for $\ind_j$ 
		: $\fprj = \Pr(\ind_j (x) = 1 | x \notin \mS_j)$ 
        \tabularnewline
		\hline
		$\fnrj$ & False-negative ratio for $\ind_j$ :
		$\fnrj = \Pr(\ind_j (x) = 0 | x \in \mS_j)$ \tabularnewline
		\hline
		$\mronej$
        & Probability of a miss in cache $j$ given a positive indication \tabularnewline
		\hline
        $\mrzeroj$ 
        & Probability of a miss in cache $j$ given a negative indication \tabularnewline
		\hline
		$\Ponej$ & Probability of a positive indication in indicator $\ind_j$
        \tabularnewline
		\hline
        $\accsCost_j$ & Access cost of cache $j$
        \tabularnewline
		\hline
		$\missp$  & Miss penalty \tabularnewline
		\hline
		\rule{0pt}{2ex}    
		$\costhetero$  & Cost function~\eqref{eq:def_service_cost}.\tabularnewline
		\hline
		\hline
		$k$ & Number of hash functions in the Bloom filter\tabularnewline
		\hline
        $\cacheSize_j$ & Size of cache $j$ ($\abs{\mS_j} \leq \cacheSize_j$)\tabularnewline
		\hline
		\bpe & Bits per cached element in the Bloom filter\tabularnewline
		\hline
		$B_1(t)$ & Number of '1' bits in the updated Bloom filter at time $t$\tabularnewline
		\hline
		$B_0(t)$ & Number of '0' bits in the updated Bloom filter at time $t$\tabularnewline
		\hline
		$\Delta_1(t)$ & Number of bits that are '1' in the updated Bloom filter,\\ &but '0' in the stale Bloom filter at time $t$\tabularnewline
		\hline
		$\Delta_0(t)$ & Number of bits that are '0' in the updated Bloom filter,\\ &but '1' in the stale Bloom filter at time $t$\tabularnewline
		\hline
		\hline
        $r_0$  & Number of caches with negative indication accessed \tabularnewline
		\hline
		$r_1$  & Number of caches with positive indication accessed \tabularnewline
		\hline
		$\costhomo$  & Cost function for the fully-homogeneous case~\eqref{eq:def_costhomo}\tabularnewline
		\hline
        $\argminCosthomoZero$ & Optimal choice of $r_0$~\eqref{eq:def_argmin_costhomo}.\tabularnewline
        \hline
        $\argminCosthomoOne$ & Optimal choice of $r_1$~\eqref{eq:def_argmin_costhomo}\tabularnewline
        \hline
        $\mrj$ & Probability of a miss in cache $j$ given its indication \tabularnewline
        \hline
		$\delta$ & Smoothness parameter of moving average ~\eqref{eq:ewma}\tabularnewline
		\hline
        \end{tabular}
        \normalsize
\end{table}

This section formally defines our system model and notations, which are summarized in Table~\ref{tbl:bf:notations}.

We consider a set $\SetCaches$ of $\NumCaches = |\SetCaches|$ {\em caches}, containing possibly overlapping sets of
items.
\revision{
    We denote by $[i]$ the set of integers $\set{0, \ldots, i}$.
}%
\revision{
Let $\mS_j$ denote the set of items stored at cache $j$ (at some point in time). 
Given a request for an arbitrary item $x$, we let $\Phitj$ denote the probability that $x \in \mS_j$.
This probability depends on the distribution of the requests, as well as on the cache policy.
$\Phit_j$ is commonly referred to as the {\em hit ratio}, i.e., the fraction of requests in a sequence $\sigma$ that were available in cache $j$, upon being issued.
Similarly to previous works, we assume that it is possible to produce a reasonable estimation of the hit ratio from the history~\cite{adaptive_cache,TinyLFU}.%
\footnote{We provide further details of how to obtain such an estimation in Sec.~\ref{sec:client-side-alg}.
}
Each cache $j$ maintains an {\em indicator} 
    $\ind_j$, which approximates the set of items in cache $j$.
    $\ind_j (x) = 1$ is referred to as a {\em positive indication} while $\ind_j (x) = 0$ is considered a {\em negative indication}.  

In what follows, when estimating the probabilities affecting the performance of our system, we consider these probabilities with respect to an arbitrary item $x$ being drawn from the distribution of items in the request sequence.

The {\em false-positive ratio} of $\ind_j$ is defined by 
    $\fpr_j = \Pr(\ind_j (x) = 1 | x \notin \mS_j)$.
    It captures the probability that an arbitrary request $x$ 
    that is 
    not in $\mS_j$, 
    is mistakenly marked by theh indicator as being in $\mS_j$.        
    Similarly, the {\em false-negative ratio} of $\ind_j$, defined by $\fnr_j = \Pr(\ind_j (x) = 0 | x \in \mS_j)$, is the  probability that the indicator mistakenly indicates that a request for an arbitrary $x$ is not in $\mS_j$.
    For every cache $j$, we denote by $\mrone_j = \Pr (x \notin \mS_j | \ind_j (x) = 1)$ the {\em positive exclusion probability}, that is, the probability that a requested arbitrary item $x$ is not in the cache, despite a positive indication.
    Similarly, we let $\mrzero_j = \Pr (x \notin \mS_j | \ind_j (x) = 0)$ denote the {\em negative exclusion probability}, that is, the probability that a requested arbitrary item $x$ is not in the cache, given a negative indication.

    We denote by $q_j$ the the {\em positive indication ratio}, namely, the probability of a positive indication for an arbitrary item requested from cache $j$.
    We note that if $\mronej=0$ then any access to cache $j$ upon a positive indication results in a cache hit, 
    and if $\mronej=1$ then any access to cache $j$ upon a positive indication results in a cache miss.
    We therefore assume hereafter that $0 < \mronej < 1$.

A positive indication for cache $j$ occurs when either an arbitrary requested item is in $S_j$ and no false-negative occurs; or the item is not in $S_j$ and a false-positive occurs.
}
Hence,
\begin{equation}\label{eq:Pone}
\Ponej = \Pr(\ind_{j}(x) = 1) = \Phit_{j} \cdot \left(1 - \fnrj \right) + (1 - \Phit_{j}) \cdot \fprj.
\end{equation}
Using Bayes' theorem, it follows that  
\begin{align}
\mronej  &= \Pr (x \notin \mS_j | \ind_{j} (x)= 1) = \fprj \cdot (1 - \Phitj) / \Ponej \label{eq:mrj1} \\
\mrzeroj &= \Pr (x \notin \mS_{j} | \ind_{j} (x)= 0) \notag \\
              & = \left(1 -\fprj \right) \cdot (1 - \Phitj) / (1 - \Ponej), \label{eq:mrj0}
\end{align}
for $\Ponej$ as defined in~\eqref{eq:Pone}.

We say that a system is {\em sufficiently-accurate} if for every indicator of cache $j$, $\fpr_{j} + \fnr_{j} < 1$.%
\footnote{
This definition is inspired by the notions of {\em accuracy} and {\em informedness}~\cite{powers11evaluation}.
}
We note that in most real-life scenarios, both the false-positive ratio $\fpr_{j}$, and the false-negative ratio $\fnr_{j}$, are well below 0.5, and therefore such systems are sufficiently-accurate.

The following simple condition characterizes sufficiently-accurate systems. 
\omitproofdisclaimer{Due to space constraints, some of the proofs are omitted, and available in~\cite{FN_aware_TR}.}
\begin{proposition}
\label{Prop:suffice-accurate}
A system is sufficiently-accurate iff for every $j$ it holds that $\mrzeroj > \mronej$.
\end{proposition}

\proofsketch{
    The proof of Prop.~\ref{Prop:suffice-accurate} follows from algebraic manipulation.
    Intuitively, the condition $\mrzeroj > \mronej$ implies that an item is more likely to be in the cache given a positive indication than its likelihood of being in the cache given a negative indication.
    This condition states that indications (either positive or negative) are more likely to be correct than incorrect. 
}

\fullproof{
    \begin{proof}
    Since the claim holds for any cache $j$, we omit the subscript $j$ for simplicity.
    We first note that the condition $\mrzero > \mrone$ implies that an item is more likely to be in the cache given a positive indication than its likelihood of being in the cache given a negative indication.
    Intuitively, this condition states that indications (either positive or negative) are more likely to be correct than incorrect. 
    In turn, this translates to having $\fpr < 0.5$ and $\fnr < 0.5$, which is indeed satisfied in a sufficiently-accurate system ($\fpr + \fnr < 1$).
    
    To formally prove Prop.~\ref{Prop:suffice-accurate}, we use the expressions for $\mrone$ (Eq.~\ref{eq:mrj1}) and $\mrzero$ (Eq.~\ref{eq:mrj0}), which imply that 
    $\mrzero > \mrone$
    if and only if
    \begin{equation}\label{eq:Prop:suffice-accurate}
    \frac{\left(1 - \fpr\right) \left(1 - \Phit \right)}{1 - \Pone} > 
    \frac{\fpr \left(1 - \Phit \right)}{\Pone}.
    \end{equation}
    Rearranging~\eqref{eq:Prop:suffice-accurate}, this is equivalent to having $\Pone > \fpr$. Assigning the expression for $\Pone$ (Eq.~\ref{eq:Pone}), we obtain $\Phit \left(1 - \fnr \right) + (1 - \Phit) \fpr > \fpr$, which is equivalent to $\fpr + \fnr < 1$, thus completing the proof.
    \end{proof}
}

\revision{
    Note that by our assumption that $\mronej>0$, Prop.~\ref{Prop:suffice-accurate} implies that $\mrzeroj>0$.
    For a {\em concrete} query $x$, let $\SetP$ denote the set of caches with positive indications for $x$, i.e., $\SetP = {\set{j | \ind_j(x)=1}}$, and $\numP = \abs{\SetP}$.
}
A request for such a datum $x$ triggers a {\em data access} which consists of 
\begin{inparaenum}[(i)]
\item querying for $x$ in all the $\NumCaches$ indicators,
\item selecting a subset $\subsetCaches \subseteq \SetCaches$ of caches, and \item accessing all the $|\subsetCaches|$ selected caches in parallel.
\end{inparaenum}
Accessing cache $j$ incurs some predefined {\em access cost}, $\accsCost_j$.
For ease of presentation, we assume without loss of generality that $\min_j c_j = 1$. 
The overall access costs of accessing a set $\subsetCaches$ of caches is $\accsCost_{\subsetCaches} = \sum_{j \in \subsetCaches} \accsCost_j$. 

A multi-cache data access is considered a {\em hit} if the item $x$ is found in at least one of the accessed caches, and a {\em miss} otherwise. 
A miss incurs a {\em miss penalty} of $\missp$, for some $\missp \geq 1$.

In our model, we do not assume any specific sharing policy among the caches. 
Yet, in the analysis of our system (Sections~\ref{sec:homo}-\ref{sec:hetero}) we assume that the exclusion probabilities are mutually independent. 
Under this assumption, our analysis provides a baseline for understanding the performance of such systems.
However, in the evaluation of our algorithms, we consider environments where the exclusion probabilities are not necessarily mutually independent (Sec.~\ref{sec:sim}).

\revision{
The {\em miss cost} of an access to a caches set $\subsetCaches$ captures the expected cost of a miss, namely, the miss penalty, times the probability of a miss. Formally, the miss cost for a query $x$ is
$\missp \cdot \prod_{\substack{j \in \subsetCaches\\ \ind_j(x) = 1}} \mrone {j} \cdot \prod_{\substack{j \in \subsetCaches\\ \ind_j(x) = 0}} \mrzero {j}$.
}
The (expected) {\em service cost} of a query is the sum of the access cost and the miss cost, namely, 
\revision{
\begin{equation}
\label{eq:def_service_cost}
\begin{split}
\costhetero_x (\subsetCaches)  = 
\sum\nolimits_{j \in \subsetCaches} \accsCost_j + \missp
\prod_{\substack{j \in \subsetCaches\\ \ind_j(x) = 1}} \mronej
\cdot \prod_{\substack{j \in \subsetCaches\\ \ind_j(x) = 0}} \mrzeroj.
\end{split}
\end{equation}
The {\em Cache Selection with False-Negative Awareness} (\cacheprobFna)
problem is to find a subset of caches $\subsetCaches \subseteq \SetCaches$ that minimizes the expected cost $\phi_x(\subsetCaches)$.
}%
\footnote{When clear from the context, we will omit the subscript $x$ from $\costhetero_x$.}

In what follows, we refer to an access to a cache with a positive indication as a {\em positive} access, and refer to an access to a cache with a negative indication as a {\em negative} access. 
In particular, we consider two types of approaches to solving the cache selection problem:
\begin{inparaenum}[(i)]
\item {\em false-negative oblivious ($\fno$)} schemes, which only perform positive accesses, and
\item {\em false-negative aware ($\fna$)} schemes, which may also perform negative accesses.
\end{inparaenum}
While the former may be viewed as the traditional way access strategies are designed, the latter is a more speculative approach, which sometimes accesses a cache even with no positive indication, risking an increased access cost.

\revision{
    We say that an algorithm Alg is an $\alpha$-approximation for the \cacheprobFna\ problem if the service cost incurred by any algorithm is at least $1/\alpha$ times the expected service cost incurred by Alg.
}

\section{Bloom Filters and Staleness}
\label{sec:bloom_filters}

This section considers Bloom filters~\cite{Bloom}, which we use as a primary example of indicators.
In particular, we will discuss how staleness affects false-negatives and false-positives in such indicators.
We note that our approach described hereafter can also be applied to other types of indicators (e.g., TinyTable~\cite{TinyTableJournal}).%
\footnote{
In general, the false-negative and false-positive ratios are strongly related to the structure of the indicator, and at times even to its specific implementation.}

A Bloom filter is a randomized data structure that approximately represents a set of items.
A Bloom filter $\ind$ consists of a bit array of size $\abs{\ind}$, and $k$ {\em independent} hash functions.
When adding an item to the filter, each of the $k$ hash functions is applied to the item, and the corresponding bit in the array is set.
When testing for an item's existence, we apply the $k$ hash functions and test the corresponding bits. If all bits are set, the Bloom filter replies with a positive indication. Otherwise, the indication is negative.

\revision{
Figure~\ref{fig:Bloomexample} illustrates the operation of a Bloom filter with two hash functions.
The only items inserted into the Bloom filter are X and Y.
X and Y were added to the Bloom filter by setting the hashes of their keys, colored by blue (X), and green (Y).
To query for X, one should apply the two hash functions, and test the corresponding values in the indicator.
As these two (blue) bits are set, the indication is true positive.
Similarly, a query for Y tests the two green bits, also resulting in a true positive indication. 
Consider next item Z, which was not inserted into the filter.
The two hashes of Z are depicted in orange.
As both of them collide with bits that were set by the insertion of X and Y, the indication is a {\em false positive} -- namely, a positive indication for a datum Z that was not inserted into the filter.
Finally, consider a query for the item W, which was not inserted into the Bloom filter.
One of the hashes of W collides with the hash of X, while the other does not exhibit a collision, and is therefore reset. Hence, the indication for W is a true negative.
}

\begin{figure}[th]
    \includegraphics[width=\columnwidth]{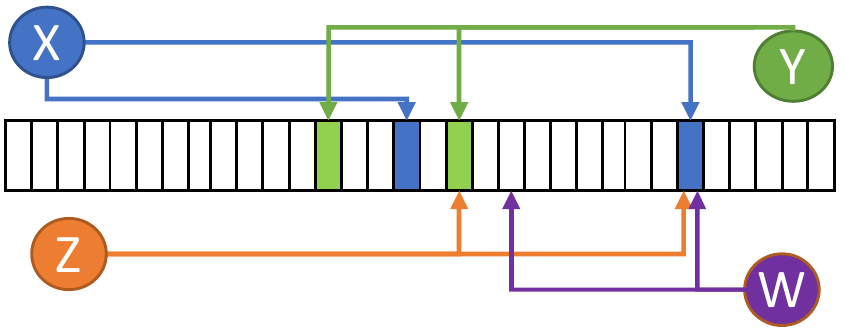}
    \caption{\label{fig:Bloomexample}
    \revision{
        An example of a Bloom filter with two hash functions. A white box represents a reset bit, while a colored box represents a set bit. 
    }
    }
\end{figure}

In a fresh Bloom filter, which is updated upon every insertion of an item to the set, positive indications may be false due to hash collisions, whereas negative indications are guaranteed to be correct.
However, in a stale Bloom filter, negative indications may also be erroneous.
Such false-negatives occur, e.g., when indicators are advertised to the clients only periodically. In such a scenario, when a new item is admitted, but the updated indicator is not yet advertised, the stale indicator available to the client fails to represent this change.

To allow a meaningful analysis of the trade-off between accuracy and memory footprint, it is useful to express the size of Bloom filters using the notion of {\em Bits Per Element (\bpe)}.
Intuitively, optimally configured Bloom filters of the same \bpe\ have the same false-positive accuracy, regardless of the size of the set being approximated by the Bloom filter.
More formally, given the value of \bpe, one can calculate the optimal number of hash functions $k$, minimizing the false-positive ratio ~\cite{Survey12}. 

In the context of caches,  let $\cacheSize_j$ denote the {\em size} of cache $j$, i.e., the maximum number of elements that can be stored in cache $j$.
The  size of the bloom filter indicator $\ind_j$ associated with cache $j$ is therefore $\abs{\ind_j} = \mbpe \cdot \cacheSize_j$.
Each cache manages its own Bloom filter, and occasionally advertises the indicator to the client.
At any time $t$, we consider the {\em updated} Bloom filter maintained by the cache, and let $B_1(t)$ and $B_0(t)$ denote the number of bits set (i.e., with value 1) and reset (i.e., with value 0), at time $t$, respectively, in the Bloom filter approximating the content of the cache.
A client uses a replica of the cache's Bloom filter, representing a snapshot of the Bloom filter that the cache advertised at some time $t' \leq t$.
We refer to this replica as the {\em stale} Bloom filter.
Let $\Delta_1(t)$ denote the number of bits that are set in the updated Bloom filter but are reset in the stale Bloom filter.
Similarly, we let $\Delta_0(t)$ denote the number of bits that are reset in the updated Bloom filter, but are set in the stale Bloom filter. Fig.~\ref{fig:calc_fnr_toy_example} illustrates this situation.
For clarity of presentation, and
\revision{without loss of generality,}
we group all the bits contributing to $\Delta_0(t)$, and group the bits accounted by $\Delta_1(t)$.

\begin{figure}
 	\footnotesize
 	\centering
    \begin{tabular}{c | c | c | c | c |}
	\cline{2-5}
	 & 
	 \multicolumn{2}{c|}{$B_0(t)$} 
	 & 
	 \multicolumn{2}{c|}{$B_1(t)$} \\
	\cline{2-5}
 	 updated BF& $0 \dots 0$ & $0 \dots \dots \dots 0$ & 1 \dots 1 & $1 \dots 1$ \\
	\cline{1-5}
 	 stale BF& $1 \dots 1$ & $0 \dots \dots \dots 0$ & 0 \dots 0 & $1 \dots 1$ \\
	\cline{2-5}
	\rule{0pt}{2ex}    
 	 & $ \Delta_0(t)$ & $B_0(t) - \Delta_0(t)$ & $\Delta_1(t)$ & $B_1(t) - \Delta_1(t)$ \\
	\cline{2-5}
    \end{tabular}
	\normalsize
	\caption {An example of an updated and a stale Bloom filter at time $t$.}
	\label{fig:calc_fnr_toy_example}
\end{figure}

\paragraph*{The false-negative ratio}
Consider a query for an item $x$ that is stored in the cache at time $t$.
Recall that an updated Bloom filter never exhibits false-negatives.
Hence we know that all the $k$ hashes of $x$ are mapped to the bits that are set in the updated Bloom filter.
The query for $x$ in a stale indicator is a true positive iff all the hashes are mapped
to the set of $B_1(t) - \Delta_1(t)$ bits that are also set in the stale indicator;
by the fact that the $k$ hash functions are independent, and uniformly distributed over their range, 
this happens with probability $\left[ \frac{B_1(t) - \Delta_1(t)}{B_1(t)} \right]^k$.
Otherwise, the query for $x$ is a false-negative.
It follows that the false-negative ratio of the cache at time $t$ can be estimated by
\begin{equation}\label{eq:estimate_fnr}
\fnr_t = 1 - \left[ \frac{B_1(t) - \Delta_1(t)}{B_1(t)} \right]^k. 
\end{equation}

\paragraph*{The false-positive ratio}
Consider a query for an item $y$ that is not stored in the cache at time $t$.
For uniformly distributed and independent hash functions, 
the hashes of $y$ are mapped to arbitrary locations in the Bloom filter.
The stale Bloom filter exhibits a false-positive iff all the $k$ hashes of $y$ map to bits that are set in the stale Bloom filter.
Hence, the probability of a false-positive in the stale indicator at time $t$ can be estimated by:
\begin{equation}\label{eq:estimate_fpr}
\fpr_t = \left[ \frac{B_1(t) - \Delta_1(t) + \Delta_0(t)}{|I_j|}\right]^k. 
\end{equation}


We note that ~\eqref{eq:estimate_fnr} and~\eqref{eq:estimate_fpr} are only estimations, while the exact miss probabilities strongly depend upon the workload, and the cache policy.
For instance, consider the case where immediately after cache $j$ sends an update, it caches an item $x$.
Then, any subsequent requests for $x$, until the cache advertises the next update, are false-negatives.
However, until the cache advertises the next updated indicator, $x$ may be accessed many times, or not accessed at all, according to the concrete workload.

To conclude this section, we note that one may apply the above analysis to more complex filters. One such example, which we use in Sec.~\ref{sec:sim}, is a compressed {\em Counting Bloom Filter} (CBF)~\cite{CBF}, which behaves similarly to a simple BF, while also supporting the removal of objects.



\section{The Fully Homogeneous Case}
\label{sec:homo}

In this section, we focus on a simplified fully-homogeneous case.
In such settings, the access cost of all caches is the same, and is normalized to one ($c=1$). The per-cache hit ratio, false-positive ratio, and false-negative ratio, are identical for all caches. I.e., for each cache $j$, $\Phitj = \Phit$, $\fprj
= \fpr$ and $\fnrj = \fnr$ for some constants $\Phit, \fpr, \fnr \in [0,1]$. 

We use this homogeneous setting to
explore the challenges and potential benefits arising from developing a false-negative aware cache selection strategy.
We first describe the aspects specific to such homogeneous settings, and then describe and analyze our false-negative aware {\em Homogeneous Cache Selection} policy, \ecmfna.
Our analysis shows that \ecmfna\ minimizes the service cost in the fully-homogeneous case.
Later on, we use \ecmfna\ to derive insights as to when it is beneficial to access a cache despite a negative indication.

\nonjrnlonly{\textbf{Preliminaries:}}
\jrnlonly{\subsection{Preliminaries}\label{sec:homo:preliminaries}}
In fully-homogeneous settings,
\revision{
all the caches with positive indication have the same access cost, and positive exclusion probability. Hence, once one decides on {\em how many} caches with positive indication to access, any choice of the concrete identities of the selected caches incurs the same expected cost. This also holds for the caches with negative indication. The task of selecting a subset of the caches $D \subseteq \SetCaches$ that minimizes the service cost is thus reduced to selecting two integers:
$r_1 \in [\numP]$,
the number of caches with positive indication to access; and
$r_0 \in [\NumCaches - \numP]$,
the number of caches with negative indication to access. The objective function $\costhetero$ ~\eqref{eq:def_service_cost} is reduced to 
\begin{equation}\label{eq:def_costhomo}
\costhomo (r_0, r_1) = 
r_0 + r_1 + \missp \cdot \mrzero^{r_0} \cdot \mrone^{r_1}.
\end{equation}

We let $\argminCosthomoZero$ and $\argminCosthomoOne$ denote the values of $r_0$ and $r_1$ that minimize the service cost, namely

\begin{equation}\label{eq:def_argmin_costhomo}
\costhomo (\argminCosthomoZero, \argminCosthomoOne) = \min_{ \substack{r_1 \in [\numP] \\ r_0 \in [\NumCaches - \numP]}} \costhomo (r_0, r_1).
\end{equation}
We assume without loss of generality that $\argminCosthomoOne$ is the maximal $r_1 \in [\numP]$ satisfying~\eqref{eq:def_argmin_costhomo}, and that $\argminCosthomoZero$ is the maximal $r_0 \in [\NumCaches - \numP]$ subject to  $\argminCosthomoOne$.
} 

\vspace{0.2cm}

\subsection{An Optimal Strategy}\label{sec:alg_best_practice}
Our algorithm for the fully-homogeneous settings, \ecmfna, is defined in Algorithm~\ref{alg:best_practice}. 
The algorithm first calculates the number of caches with positive indication to access, $\argminCosthomoOne$, assuming no cache with negative indication is accessed (Line~\ref{alg:best_practice:line_k1}). 
Next, if the
expected miss cost is still higher than accessing an additional cache
(the condition in Line~\ref{alg:best_practice:if}), the algorithm also considers caches with negative indications (Line~\ref{alg:best_practice:line_k0}).

\begin{algorithm}[t!]
    \revision{
    	\caption{\ecmfna} \label{alg:best_practice}
    	\begin{algorithmic}[1]
    		\State $\tilde{r}_0 = 0;  \tilde{r}_1 =
    		\max 
    		\argmin_{r_1 \in { [\numP] } } \left[r_1 + \missp \cdot \mrone^{r_1}\right]$\label{alg:best_practice:line_k1}
    		\Statex
    		\Comment{$\tilde{r}_1$ is the maximum value obtaining the minimum}
    		\If {$\missp \cdot \mrone^{\tilde{r}_1} > 1$}
    		\label{alg:best_practice:if}
    		\State $\tilde{r}_0 = \max \argmin_{r_0 \in  [\NumCaches - \numP]} \left[r_0 + \tilde{r}_1 + \missp \cdot \mrone^{\tilde{r}_1}\mrzero^{r_0}\right]$\label{alg:best_practice:line_k0}
    		\Statex
    		\Comment{$\tilde{r}_0$ is the maximum value obtaining the minimum subject to $\tilde{r}_1$}
    		\EndIf
    		\State \Return {$\tilde{r}_0, \tilde{r}_1$}
    	\end{algorithmic}
    }%
\end{algorithm}

\revision{
The following lemma characterizes an optimal solution in the case where not all caches with positive indications are accessed.

\begin{lemma}\label{lemma:if_r1_star_leq_n_p}
If  a  fully-homogeneous  system  is  sufficiently-accurate and $\argminCosthomoOne < \numP$, then $\argminCosthomoZero=0$.
\end{lemma}
\begin{proof}
Assume that $\argminCosthomoOne < \numP$, and assume by contradiction that $\argminCosthomoZero>0$. It follows that 
\begin{align}
\label{eq:lemma:if_r1_star_leq_n_p}
& \costhomo (\argminCosthomoZero, \argminCosthomoOne) - \costhomo (\argminCosthomoZero-1, \argminCosthomoOne+1) = \notag \\ 
& \argminCosthomoZero + \argminCosthomoOne + \missp \mrzero^{\argminCosthomoZero} \mrone^{\argminCosthomoOne} -
\left[\argminCosthomoZero-1 + \argminCosthomoOne+1 + \missp \mrzero^{\argminCosthomoZero-1} \mrone^{\argminCosthomoOne+1} \right] = \notag \\ 
& \missp \mrzero^{\argminCosthomoZero-1} \mrone^{\argminCosthomoOne} (\mrzero - \mrone) > 0,
\end{align}
where the last inequality holds true by Proposition~\ref{Prop:suffice-accurate} (recall that the system is sufficiently-accurate).
Eq.~\ref{eq:lemma:if_r1_star_leq_n_p} implies that $\costhomo (\argminCosthomoZero, \argminCosthomoOne) > \costhomo (\argminCosthomoZero-1, \argminCosthomoOne+1)$, thus contradicting the definitions of $\argminCosthomoZero$ and  $\argminCosthomoOne$~\eqref{eq:def_argmin_costhomo}.
\end{proof}
}

The following theorem shows that \ecmfna\ is optimal in the fully-homogeneous case.%

\begin{theorem}
\label{thm:ecm_optimality}
If a fully-homogeneous system is sufficiently-accurate, then \ecmfna\ minimizes the service cost. 
\end{theorem}
\revision{
\begin{proof}
First, consider the case where $\argminCosthomoOne < \numP$. 
By Lemma~\ref{lemma:if_r1_star_leq_n_p}, if $\argminCosthomoOne < \numP$, then
$\argminCosthomoZero=0$, and therefore 
\begin{align}
\costhomo (\argminCosthomoZero, \argminCosthomoOne) = \costhomo(0, \argminCosthomoOne)=  \argmin_{r_1 \in [\numP]} \left[r_1 + \missp \cdot \mrone^{r_1}\right]. \notag
\end{align} 
Recall that $\argminCosthomoOne$ is the maximal $r_1 \in [\numP]$ minimizing~\eqref{eq:def_argmin_costhomo}. On the other hand, in  Line~\ref{alg:best_practice:line_k1} \ecmfna\ assigns $\tilde{r}_1$ to the maximal
$r_1 \in [\numP]$
that minimizes $r_1 + \missp \mrone^{r_1}$.
It follows that  $\argminCosthomoOne = \tilde{r}_1$.
Hence, we can substitute $\tilde{r}_1$ by $\argminCosthomoOne$ in Line~\ref{alg:best_practice:line_k0} of the algorithm, thus obtaining
\begin{align}\label{eq:thm:ecm_optimality:r0_tilde_eq_argmincosthomoZero}
\tilde{r}_0 = \max \argmin_{r_0 \in  [\NumCaches - \numP]} \left[r_0 + \argminCosthomoOne + \missp \cdot \mrone^{\argminCosthomoOne}\mrzero^{r_0}\right].
\end{align}
Since both $\tilde{r}_0$ and $\argminCosthomoZero$ are the maximal values minimizing~\eqref{eq:thm:ecm_optimality:r0_tilde_eq_argmincosthomoZero} (subject to $\argminCosthomoZero$), it follows that 
$\tilde{r}_0 = \argminCosthomoZero$, which implies that $\tilde{r}_0, \tilde{r}_1$ minimize $\costhomo$. 

Next, assume that $\argminCosthomoOne=\numP$. Then, for every integer $r_1 \in [\numP]$, we have $\costhomo(\argminCosthomoZero, \numP) \leq \costhomo(\argminCosthomoZero, r_1)$.
Assigning this in the definition of $\costhomo$~\eqref{eq:def_costhomo}, we have
\begin{align}\label{eq:costhomo_nx_leq_costhomo_r_1}
\argminCosthomoZero + \numP + \missp \mrzero^{\argminCosthomoZero} \mrone^{\numP} \leq 
\argminCosthomoZero + r_1 + \missp \mrzero^{\argminCosthomoZero} \mrone^{r_1}.
\end{align}
Rearranging~\eqref{eq:costhomo_nx_leq_costhomo_r_1}, we obtain
\begin{align}\label{eq:costhomo_nx_leq_costhomo_r_1_rear}
\numP - r_1 \leq \missp \mrzero^{\argminCosthomoZero} \mrone^{r_1} (1 - \mrone^{\numP-r_1}) \leq
\missp \mrone^{r_1} (1 - \mrone^{\numP-r_1}),
\end{align}
where the last inequality holds true because $0 < \mrzero \leq 1$ and $\argminCosthomoZero \geq 0$. 
Reformulating~\eqref{eq:costhomo_nx_leq_costhomo_r_1_rear}, we obtain
\begin{align}\label{thm:ecmfna:argmincosthomoOne_eq_Nx}
\numP + \missp \mrone^{\numP} \leq r_1 + \missp \mrone^{r_1}.    
\end{align}
By~\eqref{thm:ecmfna:argmincosthomoOne_eq_Nx} and the maximiality of $\tilde{r}_1$ in  Line~\ref{alg:best_practice:line_k1} of  Algorithm~\ref{alg:best_practice}, \ecmfna\ assigns $\tilde{r}_1 = \numP = \argminCosthomoOne$. 
Similarly to the case where $\argminCosthomoOne<\numP$, $\tilde{r}_0$ is set to $\argminCosthomoZero$ (using the same argument as in Eq.~\ref{eq:thm:ecm_optimality:r0_tilde_eq_argmincosthomoZero}), thus completing the proof.
\end{proof}
}

\jrnlonly{
    \subsection{Quantifying the Benefits of False-Negative-Awareness} \label{Sec:homo:illust}
    We now study the potential benefits of false-negative-awareness. 
    We do so by calculating the expected service cost of the false-negative-aware approach and the cost of the traditional,  false-negative-oblivious approaches. To make our quantitative comparison meaningful, we will also be considering a {\em Perfect Indicator}, namely, an indicator with $\fpr = \fnr = 0$. 
    
    \revision{
    The expected service cost of any indicator-based access strategy depends upon the number of positive indications $\numP$.
    In this section, we consider this value as a random variable, whose distribution depends on the distribution of the items being requested.
    We, therefore, start by calculating the distribution of $\numP$.
    }
    We interpret each positive indication as the outcome of an independent Bernoulli trial with a success probability of $\Pone$ 
    (the probability of a positive indication).
    Hence, $\numP$ follows the Binomial distribution: 

    \begin{equation}
    \label{eq:dist_of_n_p}
    \Pr (\numP =j) = \binom{\NumCaches}{j} \Pone^j (1- \Pone)^{\NumCaches - j},
    \end{equation}
    and the service cost is
    \begin{equation}
    \label{eq:tciif} \costhomo = \sum_{j = 0}^{\NumCaches} \left[ \binom{\NumCaches}{j}
    \Pone^j (1- \Pone)^{\NumCaches - j} \costhomo \left (\argminCosthomoZero(j), \argminCosthomoOne(j)\right) \right],
    \end{equation}
    where $\argminCosthomoZero(j)$ and $\argminCosthomoOne(j)$ are calculated by \ecmfna.
    
    We now calculate the concrete value of~\eqref{eq:tciif} for the special case of a {\em  perfect indicator} (\pif), namely, an indicator that has neither false-positives, nor false-negatives. 
    As the \pif\ has no false indications, expected cost minimization is obtained by accessing either a 
    single cache, when there exists at least one positive indication; or no cache, otherwise. The probability of the latter case is $\Pmiss^\NumCaches$, while the probability of the first case is $1 - \Pmiss^\NumCaches$. Hence,
    \begin{equation}
    \label{eq:tcpif}
    \tcpif = 1 - \Pmiss^\NumCaches  + 
    \missp \cdot 
    \Pmiss^\NumCaches = 1 + (\missp - 1) \Pmiss^\NumCaches
    \end{equation}
    
    Note that~\eqref{eq:tcpif} can also be derived as a special case of~\eqref{eq:tciif} where $\fpr = \fnr = 0$.

    We consider two flavors of running \ecmfna, which differ by the way they take false-negative indications into account:
    \begin{inparaenum}[(i)]
    \item \ecmfna\ as defined in Algorithm~\ref{alg:best_practice}, and
    \item A traditional algorithm (\ecmfno) that essentially only performs the first line of Algorithm~\ref{alg:best_practice}, thus never accessing caches with negative indications.
    \end{inparaenum}
    We note that \ecmfno\ employs the standard approach used in caching systems that use indicators~\cite{summary_cache, Digest, Accs_Strategies_ToN}.
    
    We now quantify the potential benefits of our approach using some numerical examples. Our evaluation here is based solely on the cost analysis and equations. Namely, it does not assume any specific dataset or cache policy. In particular, our illustration provides further insight, which complements our results presented in Section~\ref{sec:alg_best_practice}.

    We consider a system with three caches and miss penalty $\missp = 100$ and focus our attention on the {\em  service costs} of the various policies, normalized by the service cost of the Perfect Indicator (\pif) configuration.
    The costs of the policies are calculated using~\eqref{eq:tciif} and~\eqref{eq:tcpif}.
        \begin{figure}
    	\centering
    
    \pgfplotsset{scaled x ticks=false}
    \pgfplotsset{scaled y ticks=false}
    
    \newcommand{\metamin}{1}
    \newcommand{\metamax}{1.17}
    \newcommand{\NumSmpls}{(\metamax-\metamin)*100 + 1} 
    \newcommand{\heatmapfontsize}{\scriptsize}
    \newcommand{\heatmapwidth}{0.5*\columnwidth}
    
    \begin{tikzpicture}
    \begin{groupplot}[
    group style = {
        group size = 2 by 1,
        horizontal sep = 10pt,
    },
    view = {0}{90},
    xtick = {0, 0.005, 0.01, 0.015, 0.02, 0.025, 0.03, 0.035, 0.04},
    ytick = {0, 0.005, 0.01, 0.015, 0.02, 0.025, 0.03, 0.035, 0.04},
    xticklabels ={0,, 0.01,, 0.02,, 0.03,, 0.04},
    yticklabels ={0,, 0.01,, 0.02,, 0.03,, 0.04},
    ylabel={False Negative Ratio ($\fnr$)},
    ylabel near ticks,
    xlabel={False Positive Ratio ($\fpr$)},
    xlabel near ticks,
    /pgfplots/colormap={blackwhite}{    color(0cm)=(blue); color(1cm)=(yellow); color(2cm)=(orange); color(3cm)=(red)},
    ]
    
    \nextgroupplot[
    font = \heatmapfontsize,
    title= \ecmfna,
    colorbar right,
    colorbar sampled, 
    width=\heatmapwidth,
    colorbar style = {
        view={0}{90},
        samples = \NumSmpls, 
        at={(6.2cm,0cm)},
        anchor=south,
        width=2mm,
        point meta min=\metamin,  
        point meta max=\metamax, 
        font = \heatmapfontsize,
        ylabel = Normalized Service Cost,
    }
    ]
    \addplot3[
    surf,
    point meta min=\metamin,
    point meta max=\metamax,
    ]
              table[header=false] {

        0.0000 0.0000 1.0000
        0.0050 0.0000 1.0188
        0.0100 0.0000 1.0376
        0.0150 0.0000 1.0388
        0.0200 0.0000 1.0394
        0.0250 0.0000 1.0401
        0.0300 0.0000 1.0407
        0.0350 0.0000 1.0414
        0.0400 0.0000 1.0422
        0.0450 0.0000 1.0430
        
        0.0000 0.0050 1.0139
        0.0050 0.0050 1.0327
        0.0100 0.0050 1.0514
        0.0150 0.0050 1.0525
        0.0200 0.0050 1.0531
        0.0250 0.0050 1.0537
        0.0300 0.0050 1.0544
        0.0350 0.0050 1.0551
        0.0400 0.0050 1.0559
        0.0450 0.0050 1.0567
        
        0.0000 0.0100 1.0280
        0.0050 0.0100 1.0468
        0.0100 0.0100 1.0654
        0.0150 0.0100 1.0663
        0.0200 0.0100 1.0668
        0.0250 0.0100 1.0674
        0.0300 0.0100 1.0679
        0.0350 0.0100 1.0685
        0.0400 0.0100 1.0692
        0.0450 0.0100 1.0698
        
        0.0000 0.0150 1.0289
        0.0050 0.0150 1.0476
        0.0100 0.0150 1.0662
        0.0150 0.0150 1.0669
        0.0200 0.0150 1.0675
        0.0250 0.0150 1.0681
        0.0300 0.0150 1.0688
        0.0350 0.0150 1.0695
        0.0400 0.0150 1.0703
        0.0450 0.0150 1.0711
        
        0.0000 0.0200 1.0292
        0.0050 0.0200 1.0479
        0.0100 0.0200 1.0666
        0.0150 0.0200 1.0673
        0.0200 0.0200 1.0680
        0.0250 0.0200 1.0688
        0.0300 0.0200 1.0696
        0.0350 0.0200 1.0705
        0.0400 0.0200 1.0713
        0.0450 0.0200 1.0723
        
        0.0000 0.0250 1.0295
        0.0050 0.0250 1.0483
        0.0100 0.0250 1.0669
        0.0150 0.0250 1.0677
        0.0200 0.0250 1.0686
        0.0250 0.0250 1.0695
        0.0300 0.0250 1.0704
        0.0350 0.0250 1.0714
        0.0400 0.0250 1.0725
        0.0450 0.0250 1.0735
        
        0.0000 0.0300 1.0298
        0.0050 0.0300 1.0487
        0.0100 0.0300 1.0672
        0.0150 0.0300 1.0682
        0.0200 0.0300 1.0692
        0.0250 0.0300 1.0702
        0.0300 0.0300 1.0713
        0.0350 0.0300 1.0724
        0.0400 0.0300 1.0736
        0.0450 0.0300 1.0748
        
        0.0000 0.0350 1.0301
        0.0050 0.0350 1.0491
        0.0100 0.0350 1.0675
        0.0150 0.0350 1.0686
        0.0200 0.0350 1.0697
        0.0250 0.0350 1.0709
        0.0300 0.0350 1.0721
        0.0350 0.0350 1.0734
        0.0400 0.0350 1.0747
        0.0450 0.0350 1.0760
        
        0.0000 0.0400 1.0304
        0.0050 0.0400 1.0495
        0.0100 0.0400 1.0678
        0.0150 0.0400 1.0691
        0.0200 0.0400 1.0703
        0.0250 0.0400 1.0716
        0.0300 0.0400 1.0730
        0.0350 0.0400 1.0744
        0.0400 0.0400 1.0758
        0.0450 0.0400 1.0773
        
        0.0000 0.0450 1.0307
        0.0050 0.0450 1.0499
        0.0100 0.0450 1.0681
        0.0150 0.0450 1.0695
        0.0200 0.0450 1.0709
        0.0250 0.0450 1.0724
        0.0300 0.0450 1.0739
        0.0350 0.0450 1.0754
        0.0400 0.0450 1.0770
        0.0450 0.0450 1.0786

                };
    \nextgroupplot [
        font = \heatmapfontsize,
        title= \ecmfno,
        yticklabels=\empty,
        ylabel=\empty,
        width=\heatmapwidth]
    \addplot3[
    surf,
    point meta min=\metamin,
    point meta max=\metamax,
    ]
    table[header=false] {
    
        0.0000 0.0000 1.0000
        0.0050 0.0000 1.0188
        0.0100 0.0000 1.0376
        0.0150 0.0000 1.0388
        0.0200 0.0000 1.0394
        0.0250 0.0000 1.0401
        0.0300 0.0000 1.0407
        0.0350 0.0000 1.0414
        0.0400 0.0000 1.0422
        0.0450 0.0000 1.0430
        
        0.0000 0.0050 1.0139
        0.0050 0.0050 1.0327
        0.0100 0.0050 1.0514
        0.0150 0.0050 1.0525
        0.0200 0.0050 1.0531
        0.0250 0.0050 1.0537
        0.0300 0.0050 1.0544
        0.0350 0.0050 1.0551
        0.0400 0.0050 1.0559
        0.0450 0.0050 1.0567
        
        0.0000 0.0100 1.0280
        0.0050 0.0100 1.0468
        0.0100 0.0100 1.0654
        0.0150 0.0100 1.0663
        0.0200 0.0100 1.0669
        0.0250 0.0100 1.0675
        0.0300 0.0100 1.0682
        0.0350 0.0100 1.0689
        0.0400 0.0100 1.0697
        0.0450 0.0100 1.0705
        
        0.0000 0.0150 1.0423
        0.0050 0.0150 1.0609
        0.0100 0.0150 1.0796
        0.0150 0.0150 1.0803
        0.0200 0.0150 1.0809
        0.0250 0.0150 1.0815
        0.0300 0.0150 1.0822
        0.0350 0.0150 1.0829
        0.0400 0.0150 1.0836
        0.0450 0.0150 1.0844
        
        0.0000 0.0200 1.0566
        0.0050 0.0200 1.0753
        0.0100 0.0200 1.0938
        0.0150 0.0200 1.0944
        0.0200 0.0200 1.0950
        0.0250 0.0200 1.0956
        0.0300 0.0200 1.0962
        0.0350 0.0200 1.0970
        0.0400 0.0200 1.0977
        0.0450 0.0200 1.0985
        
        0.0000 0.0250 1.0711
        0.0050 0.0250 1.0897
        0.0100 0.0250 1.1081
        0.0150 0.0250 1.1086
        0.0200 0.0250 1.1092
        0.0250 0.0250 1.1098
        0.0300 0.0250 1.1105
        0.0350 0.0250 1.1112
        0.0400 0.0250 1.1120
        0.0450 0.0250 1.1128
        
        0.0000 0.0300 1.0858
        0.0050 0.0300 1.1043
        0.0100 0.0300 1.1224
        0.0150 0.0300 1.1230
        0.0200 0.0300 1.1236
        0.0250 0.0300 1.1242
        0.0300 0.0300 1.1249
        0.0350 0.0300 1.1256
        0.0400 0.0300 1.1263
        0.0450 0.0300 1.1271
        
        0.0000 0.0350 1.1006
        0.0050 0.0350 1.1191
        0.0100 0.0350 1.1370
        0.0150 0.0350 1.1375
        0.0200 0.0350 1.1381
        0.0250 0.0350 1.1387
        0.0300 0.0350 1.1394
        0.0350 0.0350 1.1401
        0.0400 0.0350 1.1409
        0.0450 0.0350 1.1417
        
        0.0000 0.0400 1.1155
        0.0050 0.0400 1.1340
        0.0100 0.0400 1.1516
        0.0150 0.0400 1.1522
        0.0200 0.0400 1.1527
        0.0250 0.0400 1.1534
        0.0300 0.0400 1.1540
        0.0350 0.0400 1.1548
        0.0400 0.0400 1.1555
        0.0450 0.0400 1.1563
        
        0.0000 0.0450 1.1306
        0.0050 0.0450 1.1490
        0.0100 0.0450 1.1664
        0.0150 0.0450 1.1670
        0.0200 0.0450 1.1676
        0.0250 0.0450 1.1682
        0.0300 0.0450 1.1688
        0.0350 0.0450 1.1696
        0.0400 0.0450 1.1703
        0.0450 0.0450 1.1711

    };
    \end{groupplot}
    \end{tikzpicture}
    
        \caption{
        \label{Fig:homo_heatmap_fp_fn}
        Normalized service cost obtained by \ecmfna\ and \ecmfno\ when varying the false-positive rate and the false-negative rate.
        The number of caches is $\NumCaches=3$, the miss penalty is $\missp=100$, and the per-cache hit ratio is $\Phit=0.5$. 
        }
    \end{figure}
        \begin{figure}
    \begin{tikzpicture}
        \begin{groupplot}[%
            group style = {group size = 2 by 1, horizontal sep = 10pt},
            width = 0.59\columnwidth,
            height = 4.0cm,
    	    label style={font=\scriptsize},
    	    tick label style={font=\scriptsize},
    		xlabel= {\scriptsize Per Cache Hit Ratio},
    		xlabel near ticks,
    		xtick={0,0.2,0.4,0.6,0.8,1},
            xmin = 0,
            xmax = 1,
            ymin = 1,
            ymax = 1.5,
    	    ymajorgrids=true,
    	    grid style=dashed,
            ]
            \nextgroupplot[
                title = {\scriptsize $\fpr = 0.01, \fnr =0.01$},
                legend style = {%
                    column sep = 10pt,
                    legend columns = -1,
                    legend to name = grouplegend,
                    at={(0.5,1.23)},
                    anchor=north,
                    legend columns=-1,
                    font=\scriptsize
                    },
        		ylabel= {\scriptsize Normalized Service Cost},
        		ylabel near ticks,
                ]
                
                
        		\addplot [color = blue,  mark = *, mark options = {mark size = 2, fill = blue},  line width = \plotLineWidth]
                coordinates {
                    (0.05,1.002)(0.10,1.004)(0.15,1.007)(0.20,1.010)(0.25,1.014)(0.30,1.019)(0.35,1.026)(0.40,1.035)(0.45,1.048)(0.50,1.065)(0.
                    55,1.075)(0.60,1.086)(0.65,1.099)(0.70,1.115)(0.75,1.131)(0.80,1.140)(0.85,1.134)(0.90,1.102)(0.95,1.053)
                };
                \addlegendentry{\ecmfna}
                
                \addplot [color = black, mark = square,      mark options = {mark size = 2, fill = black}, line width = \plotLineWidth]
                coordinates {
                    (0.05,1.002)(0.10,1.004)(0.15,1.007)(0.20,1.010)(0.25,1.014)(0.30,1.019)(0.35,1.026)(0.40,1.035)(0.45,1.048)(0.50,1.065)(0.
                    55,1.080)(0.60,1.098)(0.65,1.120)(0.70,1.144)(0.75,1.167)(0.80,1.180)(0.85,1.170)(0.90,1.125)(0.95,1.060)
                };
                \addlegendentry{\ecmfno}
    
            \nextgroupplot[
                title = {\scriptsize $\fpr = 0.01, \fnr =0.05$},
                yticklabels=\empty
                ]
                

        		\addplot [color = blue,  mark = *, mark options = {mark size = 2, fill = blue},  line width = \plotLineWidth] 
                coordinates {
                    (0.05,1.008)(0.10,1.017)(0.15,1.028)(0.20,1.033)(0.25,1.035)(0.30,1.037)(0.35,1.041)(0.40,1.047)(0.45,1.056)(0.50,1.068)(0.55,1.079)(0.60,1.092)(0.65,1.106)(0.70,1.123)(0.75,1.140)(0.80,1.151)(0.85,1.143)(0.90,1.110)(0.95,1.056)
                };
                
                \addplot [color = black, mark = square,      mark options = {mark size = 2, fill = black}, line width = \plotLineWidth]
                coordinates {
                    (0.05,1.008)(0.10,1.017)(0.15,1.028)(0.20,1.040)(0.25,1.054)(0.30,1.070)(0.35,1.090)(0.40,1.114)(0.45,1.144)(0.50,1.181)(0.55,1.221)(0.60,1.267)(0.65,1.320)(0.70,1.376)(0.75,1.427)(0.80,1.447)(0.85,1.405)(0.90,1.285)(0.95,1.132)
                };
        \end{groupplot}
        \node at ($(group c2r1) + (-2.2cm,-2.5cm)$) {\ref{grouplegend}};
    \end{tikzpicture}
    \caption{
    \label{fig:homo_hit_ratio}
    Normalized service cost obtained by
    the false-negative aware Expected Cost Minimization policy (\ecmfna), and the  false-negative oblivious Expected Cost Minimization policy (\ecmfno) when varying the varying the per-cache hit ratio.
    The number of caches is $\NumCaches=3$, and the miss penalty is $\missp=100$.
    }
    \end{figure}
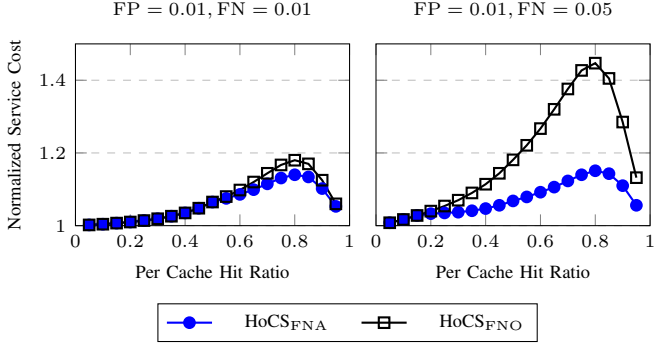
    
    Fig.~\ref{Fig:homo_heatmap_fp_fn} studies \ecmfna\ and \ecmfno, and shows the normalized service cost for different values of the false-positive ratio $\fpr$ and the false-negative ratio $\fnr$, when the per-cache hit ratio $h$ satisfies $h=0.5$.
    As one could expect, the service costs of both variants of \ecmfna\ increase with the false-positive ratio $\fpr$.
    However, for every given false-positive ratio, the service cost of the \ecmfna\ variant increases very mildly when increasing the false-negative ratio.
    In contrast, the service cost of the \ecmfno\ variant is sensitive to the false-negative ratio, and the performance degrades significantly as $\fnr$ increases.
    
    Fig.~\ref{fig:homo_hit_ratio} explores the impact of the hit ratio on the normalized service cost for a given value of $\fpr$ and $\fnr$. 
    When the per-cache hit ratio is very low (e.g., below 0.2), the caches often do not store the requested item. 
    In such a system, {\em any} indicator -- even a perfect indicator -- cannot do well. As a result, the service cost obtained by \ecmfna\ and \ecmfno\ is close to the perfect indicator (recall that the service costs are normalized w.r.t. a perfect indicator). 
    On the other extreme, when the per-cache hit ratio is very high (above 0.9), most items are stored in all the caches. In such a case, even an inaccurate indicator is likely to provide at least one true positive indication, thus bringing the expected cost again closer to that of a PI. 
    
    In the middle ground (e.g., between $\Phit=0.2$ and $\Phit=0.9$), we expect the requested datum to be available only in a small subset of the caches. we need an efficient access strategy to  
    single out this subset. 
    The left subfigure in Fig.~\ref{fig:homo_hit_ratio} shows that having a false-positive ratio and a false-negative ratio of 0.01 may induce an increase of up to $20\%$ in the normalized service cost. 
    This figure shows that when the false-negative ratio is only 0.01, the difference between \ecmfna\ and \ecmfno\ is small. 
    However, when the false-negative ratio is higher (e.g., 0.05 in the right subfigure in Fig.~\ref{fig:homo_hit_ratio}), the gap between \ecmfna\ and \ecmfno\ significantly increases, thus emphasizing the importance of a false-negative-aware approach. False-negative awareness is especially effective when the per-cache hit ratio is between 0.4 and 0.9, making it relevant for many applications.
}

\section{Dynamic and Heterogeneous Settings}
\label{sec:hetero}

In Sec.~,\ref{sec:homo} we assumed that all caches have the same false positive/negative ratios and hit ratios. In this section, we consider more realistic settings where cache attributes are {\em dynamic}, implying that they might not be accurately available at the client, and {\em heterogeneous}, i.e.,  caches may have distinct attributes.
Such an extension poses two main challenges for solving the \cacheprobFna\ problem.
\revision{
The first challenge is dealing with incomplete information, which requires estimating
the exclusion probabilities $\mronej$ and $\mrzeroj$.
These estimations guide the algorithm in its choice of caches to access. We note that such estimations should be done using the {\em limited information} available at the client-side.\footnote{\revision{In Sec.~\ref{sec:sim:estimations} we evaluate the system performance also for ideal estimations that are equipped with more information, but are impractical to implement.}}
}
This is presented in Sec.~\ref{sec:hetero:estimation}.
The second challenge is dealing with heterogeneity, which calls for efficient algorithms for choosing the set of caches to access.
This is presented in Sec.~\ref{sec:hetero:heterogeneity}.

\revision{
We dub the algorithm executed by the client for the \cacheprobFna\ problem in heterogeneous settings \pgmfna. 
The pseudo-code of \pgmfna\ appears in Algorithm~\ref{Alg:client-side}.
}
In a nutshell, the algorithm periodically obtains information from the caches
(Lines~\ref{alg:hetro:obtain_fp_and_fn}-\ref{alg:hetro:obtain_updated_indicator}), and then uses this information to estimate the exclusion probabilities for each cache (Lines~\ref{alg:hetro:estimate_positive_indication_ratio}-\ref{alg:estimate_nu_j}).
Using these values, \pgmfna\ runs a cache selection algorithm (e.g., one of the algorithms
in~\cite{Accs_Strategies_ToN}). 
We now turn to describe of our algorithm's inner workings and analysis.

\begin{algorithm}[t!]
\caption {
\revision{\pgmfna
$(\SetCaches,\Vec{\accsCost},\missp,\alg)$}}
\label{Alg:client-side} 
\begin{algorithmic}[1]
    \State periodically obtain updated $\fprj, \fnrj$ from each cache $j$%
    \label{alg:hetro:obtain_fp_and_fn}
    \State periodically obtain updated indicator $\ind_j$ from each cache $j$%
    \label{alg:hetro:obtain_updated_indicator}
    \State estimate $\Ponej$ for each cache $j$%
    \label{alg:hetro:estimate_positive_indication_ratio}
    \For {every request for datum $x$}
	    \For {$j \in [\NumCaches]$}
	       \State $\Phitj = \frac{q_j - \fnrj}{1 - \fprj - \fnrj}$%
	       \label{alg:estimate_h_j}
	       \Comment{Eq.~\eqref{eq:Pone}}
		    \If {$\ind_j(x) == 1$}%
		    \label{alg:start_generate_input_to_alg}
		    \Comment{should calculate $\mronej$}
		        \State $\mrj = \fprj \cdot (1 - \Phitj) / \Ponej$%
		        \label{alg:estimate_pi_j}
		        \Comment{Eq.~\eqref{eq:mrj1}}
		    \Else 
		    \Comment{should calculate $\mrzeroj$}
		        \State $\mrj = \left(1 -\fprj \right) \cdot (1 - \Phitj) / (1 - \Ponej)$%
		        \label{alg:estimate_nu_j}
		        \Comment{Eq.~\eqref{eq:mrj0}}
		    \EndIf
        \EndFor
        \State $D = \alg (\SetCaches, \Vec{\accsCost}, \Vec{\mr}, \missp)$%
        \label{alg:apply_reduction}
        \Comment{reduction of Theorem~\ref{thm:access_strategies_reduction}}
        \State access $D$%
        \label{alg:access_D}
    \EndFor
\end{algorithmic}
\end{algorithm}

\subsection{Estimating the Exclusion Probabilities}
\label{sec:hetero:estimation}

We now show how one can collect recent statistics of the various parameters governing system behavior, which allow the client to obtain good estimates of 
the current caches' attributes.
Our solutions use the insights presented in Sec.~\ref{sec:bloom_filters}. In particular, we will use~\eqref{eq:estimate_fnr} and~\eqref{eq:estimate_fpr} for estimating the false-positive ratios and the false-negative ratios of the distinct caches,
enabling us to compute $\mronej$, and $\mrzeroj$.

\subsubsection{Cache-side Algorithm}
\label{sec:cache-side-alg}
The cache maintains both the stale Bloom filter (i.e., the most recently advertised Bloom filter, which is also available at the client) and the updated Bloom filter.
Along a sequence of requests $\sigma$, each cache $j$ estimates the false-negative ratio
and the false-positive ratio, according to ~\eqref{eq:estimate_fnr} and~\eqref{eq:estimate_fpr}, by comparing the stale and updated Bloom filters.

We note that these estimations can be done periodically to reduce the computational overhead of comparing the stale and updated bloom filters.
These estimates are sent (periodically) to the client (Line~\ref{alg:hetro:obtain_fp_and_fn} of \pgmfna).
Each cache further (periodically) sends an updated indicator to the client (Line~\ref{alg:hetro:obtain_updated_indicator} of \pgmfna).
We note that these updates are sent in an arbitrary asynchronous manner to the client.

\subsubsection{Client-side Algorithm}
\label{sec:client-side-alg}

We now show how the client may estimate the exclusion probabilities $\mronej$ and $\mrzeroj$ for every cache $j$, given the estimations of $\fprj$ and $\fnrj$ which are periodically provided by the cache.

For evaluating $\Ponej$, the client periodically estimates the probability $\Pr(\ind_{j}(x) = 1)$ {\em empirically}, using a weighted exponential moving average.
Formally, consider a sequence of requests $\sigma$, and consider epochs of $T$ requests.
Let $a_j(s,t)$ denote the number of positive indications of indicator $\ind_j$ for requests $s+1,\ldots,t$ made 
by the client. 
For any $t \leq T$ we let the estimated positive indication ratio after handling request $t$ be $\Pone_{j,t}=\frac{a_j(0,t)}{t}$.
For every $i=1,2,\ldots$ and every $i T < t < (i+1) T$, we let $\Pone_{j,t}$ be the most recent estimate over epochs of $T$ requests, i.e., $\Pone_{j,t}=\Pone_{j,\floor{t/T}\cdot T}$,
and the estimate is updated at $t=(i+1)T$ such that
\begin{align}\label{eq:ewma}
\Pone_{j,(i+1)T}
&= \delta \cdot \frac{a_j(iT,(i+1)T)}{T} + (1-\delta) \cdot \Pone_{j,iT},
\end{align}
where $\delta\in (0,1)$ is some constant governing the dynamics of the estimate change.
We note that only the client can perform such an estimation since it requires knowing all the requests in $\sigma$, and not only requests for which the cache has been accessed.

Given the current values for $\fprj$, $\fnrj$, and $\Ponej$, for every item being requested in the sequence $\sigma$, the client estimates the hit ratio $h_j$ (Line~\ref{alg:estimate_h_j}), and the exclusion probabilities $\mronej$  (Line~\ref{alg:estimate_pi_j}) and $\mrzeroj$ (Line~\ref{alg:estimate_nu_j}) 
using~\eqref{eq:Pone}, \eqref{eq:mrj1}, and~\eqref{eq:mrj0}, respectively. These values are assigned to variables $\rho_j$, as explained in the sequel.

\subsection{Choosing the Caches to Access}
\label{sec:hetero:heterogeneity}

This section shows how to use the estimations of the exclusion probability of each cache $j$ to develop a false-negative-aware access strategy. In particular, we show how to extend any false-negative-oblivious access strategy (e.g., those in~\cite{Accs_Strategies_ToN}), to consider a non-zero false-negative ratio.

For any set of caches $D$, the client's estimations of the exclusion probabilities essentially determine the expected miss cost.
We let $\mrj$ denote the probability of a miss while accessing cache $j$, given its indication for the requested item.
Formally, $\mrj = \mronej$ if $\ind_j(x) = 1$, and $\mrj = \mrzeroj$ if $\ind_j(x) = 0$.
Then, the expected miss cost can be expressed by $\missp \cdot \prod_{j \in \subsetCaches} \mrj$, and the objective function defined in~\eqref{eq:def_service_cost} translates to finding the set of caches $D$ minimizing
\begin{align}
\label{eq:def_service_cost_access_stratgies}
\costhetero_x (\subsetCaches)
&= 
\sum\nolimits_{j \in \subsetCaches} \accsCost_j + \missp \cdot
\prod\nolimits_{j \in \subsetCaches} \mrj.
\end{align}

The problem of finding a set of caches $D$ (out of those with a positive indication) minimizing an objective of the form depicted in~\eqref{eq:def_service_cost_access_stratgies} has been studied in~\cite{Accs_Strategies_ToN}, where they present several approximation algorithms for the problem.
The problem studied in~\cite{Accs_Strategies_ToN} is essentially equivalent to assuming that there are no false-negative indications, and therefore it suffices to consider only caches for which $\ind_j(x) = 1$.
We refer to this special case 
as the
\revision{
false-negative-oblivious cache-selection problem (\cacheprobFno).
}

When considering the
\cacheprobFno\
problem within our model, the framework of~\cite{Accs_Strategies_ToN} can be viewed as assuming that {\em all} caches have a positive indication, and $\mrj$ represents the positive exclusion probability of cache $j$.
Equivalently, the model of~\cite{Accs_Strategies_ToN} essentially assumed that $\mrzeroj=1$ for all $j$, which is fundamentally not the case in the
\cacheprobFna\
problem.

Our proposed algorithm \pgmfna\ selects the set of caches to access as follows:
\begin{inparaenum}[(i)]
\item \pgmfna\ gets as input an algorithm $\alg$ for solving the
\cacheprobFno\
problem (assuming all caches have a positive indication),
\item generates the appropriate input for this algorithm (as described above) in Lines~\ref{alg:start_generate_input_to_alg}-\ref{alg:estimate_nu_j}, and
\item accesses the set of caches prescribed by algorithm $\alg$.
\end{inparaenum}
  
The following theorem serves to analyze the worst-case performance guarantees of \pgmfna.
\begin{theorem}
\label{thm:access_strategies_reduction}
If there exists an algorithm $\alg$ that is an $\alpha$-approximation algorithm for the
\cacheprobFno\
problem, then there exists an $\alpha$-approximation algorithm for the
\cacheprobFna\
problem (with arbitrary values of $\mrzeroj$).
\end{theorem}

\begin{proof}
\revision{
Assume an input to the \cacheprobFna\
problem such that every cache $j$ has its indicator $\ind_j$, and its positive and negative exclusion probabilities $\mrone_j$ and $\mrzero_j$, respectively.
In what follows we slightly abuse notation, and refer to $\costhetero_{x,\Vec{\mrone},\Vec{\mrzero},\Vec{\ind}}$ as the expected service cost for an input $x$, given these system parameters.
Let $\alg$ be an $\alpha$-approximation algorithm for the
\cacheprobFno\ problem~\eqref{eq:def_service_cost_access_stratgies}, for which its expected service cost for an input $x$ is referred to as
$\costhetero_{x,\Vec{\mrone},\Vec{\ind}}$.

Assume each cache $j$ has some arbitrary negative exclusion probability, $\mrzeroj$.
For every cache $j$, we let $\mrone_j^*=\mrone_j$ if $\ind_j(x)=1$, and let $\mrone_j^*=\mrzero_j$ if $\ind_j(x)=0$.
For every cache $j$ we define indicator $\ind_j^*$ such that $\ind_j^*(x)=1$, implying that the set of caches with a positive indication according to $\Vec{\ind}^*$ is the set of all caches, $\SetCaches$.

We define algorithm $\alg^*$ such that $\alg^*$ returns the output of $\alg$ for the inputs of $\Vec{\mrone}^*$ (for the positive exclusion probabilities), 
and the set of all caches with a positive indication according to $\Vec{\ind}^*$ (i.e., $\SetCaches$).
We now show that the solution returned by $\alg^*$ is an $\alpha$-approximate solution for the
\cacheprobFna\
problem with exclusion probabilities $\mronej$ and $\mrzeroj$.

By the assumption on $\alg$, its output $D$ satisfies
\begin{align}\label{eq:thm_proof_service_cost_fno}
\costhetero_{x,\Vec{\mrone}^*,\Vec{\ind}^*}(D) \leq \alpha \cdot \costhetero_{x,\Vec{\mrone}^*,\Vec{\ind}^*}(D^*)
\end{align}
where $D^*$ is an optimal solution to the
\cacheprobFno\
problem with $\Vec{\mrone}^*$, and the set of caches induced by $\Vec{\ind}^*$ as inputs.
By the definition of $\Vec{\mrone}^*$ and $\Vec{\ind}^*$ it follows that for every set of caches $\tilde{D}$,
\begin{align}\label{eq:thm_proof_service_cost_fna}
\costhetero_{x,\Vec{\mrone}^*,\Vec{\ind}^*}(&\tilde{D})  =
\sum_{j \in \tilde{D}} \accsCost_j + \missp
\prod_{\substack{j \in \tilde{D}\\ \ind^*_j(x) = 1}} \mrone_j^* \notag \notag \\ 
& = \sum_{j \in \tilde{D}} \accsCost_j + \missp
\prod_{j \in \tilde{D}} \Big[ \ind_j(x) \cdot \mronej + \left(1 - \ind_j(x) \right) \mrzeroj \Big] \notag \\
& = \costhetero_{x,\Vec{\mrone},\Vec{\mrzero},\Vec{\ind}}(\tilde{D}),
\end{align}
where the first equality follows from the fact that $\ind^*_j(x) = 1$ for all $j$, the second equality follows from the definition of $\mronej^*$, and the third equality follows from the definition of $\costhetero$~\eqref{eq:def_service_cost}.
Combining~\eqref{eq:thm_proof_service_cost_fno} and~\eqref{eq:thm_proof_service_cost_fna}, the output $D$ of $\alg$ satisfies
\begin{align}
     \costhetero_{x,\Vec{\mrone},\Vec{\mrzero},\Vec{\ind}}(D) & = 
    \costhetero_{x,\Vec{\mrone}^*,\Vec{\ind}^*}(D) \notag \\
    & \leq  \alpha \cdot \costhetero_{x,\Vec{\mrone}^*,\Vec{\ind}^*}(D^*) \notag \\
     & = \alpha \cdot \costhetero_{x,\Vec{\mrone},\Vec{\mrzero},\Vec{\ind}}(D^*),
\end{align}
which completes the proof.
}
\end{proof}

The proof of Theorem~\ref{thm:access_strategies_reduction} implies the following corollary.

\begin{corollary}
\label{cor:alg_approximation}
If the estimations of $\mronej$ and $\mrzeroj$ produced by \pgmfna\ are precise, and $\alg$ used by \pgmfna\ is an $\alpha$-approximation algorithm for the
\cacheprobFno\
problem, then \pgmfna\ produces an $\alpha$-approximate solution to the
\cacheprobFna\
problem.
\end{corollary}

Combining Corollary~\ref{cor:alg_approximation} with the results of~\cite{Accs_Strategies_ToN}, we obtain a myriad of trade-offs and possible approximation guarantees for \pgmfna.
In particular, in Sec.~\ref{sec:sim} we consider the performance of one specific realization of \pgmfna, which uses algorithm \dspgm\ for the
\cacheprobFno\
problem presented in~\cite{Accs_Strategies_ToN}.

\section{Simulation Study}
\label{sec:sim}
In this section, we evaluate the performance and trade-offs of our proposed false-negative aware algorithm, \pgmfna, in various scenarios, using traces of real-life workloads.
Our evaluation shows that false-negative awareness improves the oblivious approach across the board. In some cases, one needs an order of magnitude more bandwidth or more cache entries to match our false-negative aware approach's service cost. 
The effect is consistent for diverse cache sizes, workloads, and when increasing the number of caches. The difference is especially significant when the miss penalty and update interval are large. 
We begin by describing our evaluation settings and parameters.

\subsection{Simulation Settings}\label{sec:sim_settings}

\paragraph*{Traces}
We use the  first 1M requests from each of these real workload traces. 
\begin{inparaenum}[(i)]
\item {\em Wiki}: Read requests to Wikipedia pages~\cite{WikiBench}.
\item {\em Gradle}: 
Gradle is a build tool for caching compiled libraries in large projects. The trace was provided by~\cite{Scarab_and_Gradle_traces}.
\item {\em Scarab}:
A trace from Scarab Research, a personalized recommendation system for e-commerce sites~\cite{Scarab_and_Gradle_traces}.
\item {F2}: Traces from a financial transaction processing system~\cite{Umass_traces}.
\end{inparaenum}

\paragraph*{Caches}
We consider a system-wide request distribution where a missed item is placed in a single cache chosen by the controller. 
Such an approach is common in large distributed systems, such as Memcached~\cite{Memcached} and Kademlia~\cite{Kaleidoscope} for load balancing and for maximizing the cached content.

Each cache applies the \emph{Least Recently Used (LRU)} eviction policy which is arguably the most commonly used policy. 

\paragraph*{Indicators}
Each cache $j$ of size $\cacheSize_j$ periodically advertises an indicator $\ind_j$ of size $\mbpe \cdot \cacheSize_j$.
For computing the indicator, cache $j$ maintains a {\em Counting Bloom Filter} (CBF)~\cite{CBF} with 3-bit counters, where the number of counters is $\mbpe \cdot \cacheSize_j$.
The advantage of the CBF over a simple Bloom filter ~\cite{Bloom} is that the CBF supports removal of items too. 
Thus, 
we add an item to the CBF upon admission to the cache and remove an item from the CBF upon eviction. 
The cache constructs the advertised indicator by compressing the CBF to a simple (1 bit-counter) Bloom filter where a bit is set iff the respective counter in the CBF is strictly positive. 
We pick the number of hash functions that minimizes the false-positive probability~\cite{Survey18}. 

\paragraph*{Access Strategy Algorithms compared}
Recall that \pgmfna\ makes use of an algorithm for solving the \cacheprobFna\ 
problem for the case where indicators exhibit no false-negatives.
In our evaluation, we make use of the \dspgm\ algorithm from~\cite{Accs_Strategies_ToN}.
This strategy was shown to produce a $(\log \missp)$-approximation for the \cacheprobFna\ problem with no false-negatives. By Corollary~\ref{cor:alg_approximation}, this guarantee also applies to the \cacheprobFna\ problem. Furthermore, \dspgm\ exhibits close-to-optimal results in practice, when tested on real-world workloads~\cite{Accs_Strategies_ToN}. 

We consider two benchmarks for evaluating the performance of
\pgmfna:
\begin{inparaenum}[(i)]
\item applying the vanilla \dspgm\ algorithm (\pgmfno), which only considers accessing caches with a positive indication (albeit stale), using only the estimates of $\mronej$ for every cache $j$, and using $\mrzeroj=1$ for all $j$, and
\item the hypothetical ideal strategy that uses {\em perfect information} (\pif), i.e., a strategy that always has access to the precise cache content, which accesses the cheapest cache containing an item if such a cache exists, and doesn't access any cache otherwise.
\end{inparaenum}

Throughout our evaluation, both \pgmfna\ and \pgmfno\ evaluate $\Ponej$ with a time horizon of $T=100$ requests and using $\delta=0.25$ for the weighting of the moving average.
Furthermore, each cache $j$ re-estimates the false-positive ratio $\fprj$ and the false-negative ratio $\fnrj$ once every 50 insertions to the cache.


\paragraph*{Evaluation metric}
We consider the mean service cost per request over the entire input. 
We also consider the {\em normalized cost} where we divide each algorithm's  mean cost by that of the \pif\ strategy. 
While infeasible, it is instructive to use it as a lower bound on the cost of {\em any} policy for solving the \cacheprobFna\ problem.

\paragraph*{Baseline scenario} 
Unless stated otherwise, our evaluation considers three caches whose access costs are 1, 2, and 3, and a miss penalty of 100 (i.e., 50 times the average cache access cost). Each cache can store 10K elements. Similar cache sizes were considered by existing works in the field~\cite{TinyLFU,adaptive_cache}, and can further be motivated, e.g., by 
Trivago's Memcached~\cite{Trivago} that utilizes a distributed system of caches, each of size 4GB, containing items with a typical size of about 1MB.

The update interval is measured by the number of insertions. In our baseline scenario, $0.1\cdot \cacheSize_j$ insertions are performed between subsequent indicator advertisements. This translates to an advertisement once every 1K insertions for the default 10K-items cache. This is in accordance with previous work evaluating such systems~\cite{summary_cache}.
Note that periodically advertising the indicator is sometimes done once in every fixed time interval (e.g., by Squid~\cite{squidspec}).
However, the {\em optimal} time interval length strongly depends on the workload being served.
Our approach removes this dependency on the characteristics of the workload, and allows for a clearer evaluation of the effect the various system parameters have on performance, in scenarios where indicators become stale.

The advertised indicator of each cache $j$ uses $\mbpe=14$, implying an indicator size of $14\cdot \cacheSize_j$, where the number of hash functions is optimized to minimize the false-positive ratio. In particular, in our baseline scenario, this translates to a designed false-positive ratio of 0.1\%~\cite{Survey12}.
Each evaluation explores the impact of varying one of the system's parameters, where the remaining parameters are set according to our baseline scenario.
Our Python code is available in~\cite{Github_ofanan}.

\subsection{Impact of Miss Penalty and Workload Diversity}
\begin{figure}
\centering
\begin{tikzpicture}
\pgfplotstableread{three_caches.dat}{\loadedtable}

\begin{groupplot}[
    group style=
        {
        group size=3 by 1,
        xlabels at=edge bottom,
        ylabels at=edge left,
        horizontal sep=0.02\textwidth,
        vertical sep=0.06\textwidth,
        group name=plots,
        },
    ybar,
    major x tick style = transparent,
    ymajorgrids = true,
    width=0.43\columnwidth,
    height=0.37\columnwidth,
    ytick={0, 0.5, 1,1.5,2,2.5,3,3.5,4},
    yticklabel style={
        font=\footnotesize,
        text width=0.035\textwidth,
        align=right,
        inner xsep=0pt,
        xshift=-0.014\textwidth,
        },
    ylabel=Normalized service cost,
    ylabel style={
        font=\footnotesize,
    },
    symbolic x coords = {wiki,gradle,scarab,F2},
    xticklabel style={
        font=\footnotesize,
        align=center,
        yshift=0.013\textwidth,
        rotate=45,
        },
    enlarge x limits=0.2,
    xtick=data,
    ymin = 1,
    ymax= 4,
    legend style={
        legend columns=-1,
        at={(-0.645,1.7)},
        anchor=north,
        /tikz/every even column/.append style={column sep=0.3cm},
        },
    ybar,
    area legend,
    ]

\nextgroupplot[title={$\missp=50$},font=\footnotesize,]
        \addplot[ybar, draw=black, fill = white, bar width = 4pt,] table [ybar,x=input,y=FNO50]{\loadedtable};
        \addplot[ybar, draw=blue, fill = blue, bar width = 4pt,] table [ybar,x=input,y=FNA50]{\loadedtable};
\nextgroupplot[title={$\missp=100$},font=\footnotesize,yticklabels=\empty]
        \addplot[ybar, draw=black, fill = white, bar width = 4pt,] table [ybar,x=input,y=FNO100]{\loadedtable};
        \addplot[ybar, draw=blue, fill = blue, bar width = 4pt,] table [ybar,x=input,y=FNA100]{\loadedtable};
\nextgroupplot[title={$\missp=500$},font=\footnotesize,yticklabels=\empty]
        \addplot[ybar, draw=black, fill = white, bar width = 4pt,] table [ybar,x=input,y=FNO500]{\loadedtable};
        \addplot[ybar, draw=blue, fill = blue, bar width = 4pt,] table [ybar,x=input,y=FNA500]{\loadedtable};
\legend{\pgmfno,\pgmfna}
\end{groupplot}
\end{tikzpicture}
\caption{Normalized cost of the heterogeneous 3-caches baseline scenario for varying traces and miss penalty values.}
\label{fig:all_missp_and_traces}
\vspace{-0.5cm}
\end{figure}
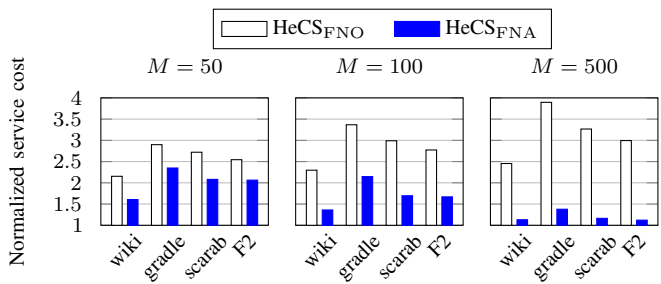

We first compare the performance of \pgmfno\ and \pgmfna\ when varying the miss penalty values $\missp$ in the range $\{50, 100, 500\}$.
The results in Fig.~\ref{fig:all_missp_and_traces} show that while the performance of the false-negative oblivious policy \pgmfno\ degrades as the miss penalty increases, the performance of our proposed false-negative aware algorithm \pgmfna\ improves significantly.
Furthermore, the performance of \pgmfna\ tends to the {\em optimal} performance as the miss penalty increases.
This behavior follows from the fact that a higher miss penalty accentuates the impact of false-negative events.
In particular, ignoring negative indications (as is done by \pgmfno) is severely penalized by an increased expected miss cost in cases where the miss penalty is large.

Fig.~\ref{fig:all_missp_and_traces} also demonstrates significant differences across distinct workloads.
\pgmfno's worst performance is exhibited for the Gradle trace, whereas its best performance is obtained for the Wiki trace.
To understand this phenomenon, we observe that Gradle exhibits a high recency-bias, where items are requested shortly after their first appearance.
As false-negatives occur when the indicator does not reflect the insertion of new items, \pgmfno, which never accesses caches with a negative indication, fails to take advantage of this recency bias.
In contrast, the Wiki trace is more frequency-biased, which implies that popular items do not rapidly change over time and that the impact of false-negatives is less pronounced.
We continue with the Wiki and Gradle traces, which are more sensitive to false negatives. 

\subsection{Impact of Advertisement Policy and Indicator Parameters}\label{sec:sim:ind_parameters}

\subsubsection{Update interval}
We now turn to study the effect of staleness on the performance of our algorithm.
To this end, we let the update interval, namely, the number of insertions between indicator advertisement, vary between 16 and 8K (8192), and consider the normalized cost of both \pgmfna\ and \pgmfno. These results are presented in Fig.~\ref{Fig:uInterval}, where we consider the performance for the Gradle and Wiki workloads.

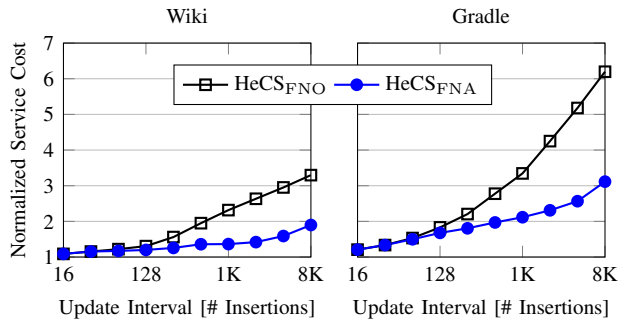
\begin{figure}[t]
    \begin{tikzpicture}
        \begin{groupplot}[
            group style=
                {
                columns = 2,
                xlabels at=edge bottom,
                ylabels at=edge left,
                horizontal sep=0.07\columnwidth,
                group name=plots
                },
    		width  = \plotWidth,
     		height = \plotHeight,
     		minor tick style={draw=none},
            ymajorgrids = true,
            ylabel = {Normalized Service Cost},
    		xmode = log,
    		xmin = 16,
     		xmax = 8192,
     		xtick =       {16, 128, 1024, 8192},
     		xticklabels = {16,128,1K,8K},
     		ymin = 1,
     		ymax = 7,
      		ytick       = {1, 2, 3, 4, 5, 6, 7},
      		yticklabels = {1, 2, 3, 4, 5, 6, 7},
     		ylabel near ticks,
    		legend style = {at={(-0.1,0.9)}, anchor=north, legend columns=-1, font=\footnotesize},
    		xlabel = {Update Interval [\# Insertions]}, 
    		xlabel near ticks,
    	    label style={font=\footnotesize},
    	    tick label style={font=\footnotesize},
    	]
        	
        \nextgroupplot[
            title = Wiki,
            title style ={
                font=\footnotesize,
                yshift = -2pt,
                },
     		]

		\addplot [color = black, mark = square,      mark options = {mark size = 2, fill = black}, line width = \plotLineWidth] coordinates {
		(16, 1.0885)(32, 1.1541)(64, 1.2214)(128, 1.3057)(256, 1.5600)(512, 1.9498)(1024, 2.3163)(2048, 2.6338)(4096, 2.9516)(8192, 3.2969)
		};		

		\addplot [color = blue,  mark = *, mark options = {mark size = 2, fill = blue},  line width = \plotLineWidth] coordinates {
		(16, 1.0885)(32, 1.1541)(64, 1.1719)(128, 1.1993)(256, 1.2555)(512, 1.3588)(1024, 1.3632)(2048, 1.4183)(4096, 1.5892)(8192, 1.8980)
		};		

        \nextgroupplot[
            title = Gradle,
            title style ={
            font=\footnotesize,
            yshift = -2pt,
                },
      		yticklabels = \empty,
     		]
     		
		\addplot [color = black, mark = square,      mark options = {mark size = 2, fill = black}, line width = \plotLineWidth] coordinates {
		(16, 1.2073)(32, 1.3354)(64, 1.5327)(128, 1.8312)(256, 2.2053)(512, 2.7762)(1024, 3.3449)(2048, 4.2503)(4096, 5.1793)(8192, 6.1963)
		};		
		\addlegendentry {\pgmfno}

		\addplot [color = blue,  mark = *, mark options = {mark size = 2, fill = blue},  line width = \plotLineWidth] coordinates {
		(16, 1.2073)(32, 1.3354)(64, 1.4983)(128, 1.6775)(256, 1.8054)(512, 1.9696)(1024, 2.1165)(2048, 2.3102)(4096, 2.5659)(8192, 3.1145)
		};		
		\addlegendentry {\pgmfna}

        \end{groupplot}2
    \end{tikzpicture}
    \caption{\label{Fig:uInterval}Normalized cost of the heterogeneous 3-caches baseline scenario for varying update intervals. Update intervals are measured by the number of cache insertions between subsequent updates.}
    \vspace{-0.5cm}
\end{figure}

Our results show that both algorithms' performance degrades as the update interval increases. When updates are relatively frequent (i.e., up to 128), the performance of \pgmfna\ and \pgmfno\ is similar. However, a significant gap emerges between the performance of both algorithms for larger update intervals.
In particular, the performance of \pgmfno, which ignores negative indications, quickly degrades, whereas \pgmfna\ shows a considerably milder degradation.
This phenomenon is directly related to the fact that when the update interval is large, the false-negative ratio increases significantly (as demonstrated in Fig.~\ref{fig:Fn_Vs_uInterval}).
Under such regimes, \pgmfno\ fails to access a cache even when the item is available at the cache, whereas \pgmfna\ relies on its false-negative awareness to make accesses even in cases of negative indications, taking into account the false-negative ratio estimation provided by the caches.
Our results imply that \pgmfna\ matches the performance of \pgmfno\ while using a significantly lower bandwidth overhead for cache advertisements.
For instance, for the Wiki workload \pgmfna\ matches the service cost as \pgmfno\ while using 16x
less bandwidth for indicator advertisements.
To see this, notice that \pgmfna's cost using an update interval of $8K$ is on par with that of \pgmfno\ with an update interval of $512$. 

\subsubsection{Indicator size}
Fig.~\ref{Fig:bpe} illustrates our results for varying the size of the indicator being used and advertised by the cache.
We vary the number of indicator bits per cached element (\bpe) and study the impact of the indicator's size on the service cost.
our evaluation compares the performance of \pgmfno\ and \pgmfna\ with update intervals of 256 and 1024.

\begin{figure}[t]
    \begin{tikzpicture}
        \begin{groupplot}[
            group style=
                {
                columns = 2,
                rows = 2,
                xlabels at=edge bottom,
                ylabels at=edge left,
                horizontal sep=0.07\columnwidth,
                group name=plots
                },
    		width  = \plotWidth,
     		height = \plotHeight,
     		minor tick style={draw=none},
            ymajorgrids = true,
            ylabel = {Normalized Service Cost},
    		xmin = 5,
     		xmax = 15,
     		xtick = {5,6,7,8,9, 10, 11, 12, 13, 14, 15},
     		xticklabels = {5,, 7,, 9,, 11,, 13,, 15},
      		ymin = 1.0,
      		ymax = 3.51,
     		ytick       = {1, 1.5, 2, 2.5, 3.0, 3.5},
     		yticklabels = {1, 1.5, 2, 2.5, 3.0, 3.5},
     		ylabel near ticks,
            legend columns = 2,    
    		legend style = {at={(-0.1, 0.87)}, anchor=north,legend columns=-1,font=\footnotesize},
    		xlabel = {Indicator Size [bits per element]}, 
    		xlabel near ticks,
    	    label style={font=\footnotesize},
    	    tick label style={font=\footnotesize},
    	]
        	
        \nextgroupplot[
            title = {Wiki, uInterval = 256},
            title style ={
            font=\footnotesize,
            yshift = -2pt,
                },
     		]
     	
		\addplot [color = black, mark = square,      mark options = {mark size = 2, fill = black}, line width = \plotLineWidth] coordinates {
		(5, 1.5643)(6, 1.5698)(7, 1.5552)(8, 1.5594)(9, 1.5689)(10, 1.5627)(11, 1.5564)(12, 1.5586)(13, 1.5613)(14, 1.5600)(15, 1.5561)
		};

		\addplot [color = blue,  mark = *, mark options = {mark size = 2, fill = blue},  line width = \plotLineWidth] coordinates {
		(5, 1.4966)(6, 1.4595)(7, 1.4487)(8, 1.3930)(9, 1.3677)(10, 1.3628)(11, 1.3327)(12, 1.2973)(13, 1.2677)(14, 1.2555)(15, 1.2405)
		};

        \nextgroupplot[
            title = {Wiki, uInterval = 1024},
            title style ={
            font=\footnotesize,
            yshift = -2pt,
                },
      		yticklabels = \empty,
     		]

		\addplot [color = black, mark = square,      mark options = {mark size = 2, fill = black}, line width = \plotLineWidth] coordinates {
		(5, 2.1211)(6, 2.1794)(7, 2.1834)(8, 2.2236)(9, 2.2598)(10, 2.2592)(11, 2.2748)(12, 2.3011)(13, 2.3274)(14, 2.3163)(15, 2.3342)
		};
		\addlegendentry {\pgmfno}

		\addplot [color = blue,  mark = *, mark options = {mark size = 2, fill = blue},  line width = \plotLineWidth] coordinates {
		(5, 1.6672)(6, 1.6392)(7, 1.5904)(8, 1.5265)(9, 1.4882)(10, 1.4697)(11, 1.4614)(12, 1.4206)(13, 1.3717)(14, 1.3632)(15, 1.3303)
		};
		\addlegendentry {\pgmfna}

        \end{groupplot}2
    \end{tikzpicture}
    
    
    \begin{tikzpicture}
        \begin{groupplot}[
            group style=
                {
                columns = 2,
                rows = 1,
                xlabels at=edge bottom,
                ylabels at=edge left,
                horizontal sep=0.07\columnwidth,
                group name=plots
                },
    		width  = \plotWidth,
     		height = \plotHeight,
     		minor tick style={draw=none},
            ymajorgrids = true,
            ylabel = {Normalized Service Cost},
    		xmin = 5,
     		xmax = 15,
     		xtick = {5,6,7,8,9, 10, 11, 12, 13, 14, 15},
     		xticklabels = {5,, 7,, 9,, 11,, 13,, 15},
      		ymin = 1.0,
      		ymax = 3.51,
     		ytick       = {1, 1.5, 2, 2.5, 3.0, 3.5},
     		yticklabels = {1, 1.5, 2, 2.5, 3.0, 3.5},
     		ylabel near ticks,
    		xlabel = {Indicator Size [bits per element]}, 
    		xlabel near ticks,
    	    label style={font=\footnotesize},
    	    tick label style={font=\footnotesize},
    	]
        	
        \nextgroupplot[
            title = {Gradle, uInterval = 256},
            title style ={
            font=\footnotesize,
            yshift = -2pt,
                },
     		]
     		
		\addplot [color = black, mark = square,      mark options = {mark size = 2, fill = black}, line width = \plotLineWidth] coordinates {
		(5, 2.1701)(6, 2.1887)(7, 2.1829)(8, 2.1842)(9, 2.2100)(10, 2.2017)(11, 2.1975)(12, 2.2120)(13, 2.2113)(14, 2.2053)(15, 2.2021)
		};

		\addplot [color = blue,  mark = *, mark options = {mark size = 2, fill = blue},  line width = \plotLineWidth] coordinates {
		(5, 2.0343)(6, 1.9772)(7, 1.9865)(8, 1.9539)(9, 1.9337)(10, 1.9020)(11, 1.8713)(12, 1.8203)(13, 1.8065)(14, 1.8054)(15, 1.7622)
		};

        \nextgroupplot[
            title = {Gradle, uInterval = 1024},
            title style ={
            font=\footnotesize,
            yshift = -2pt,
                },
      		yticklabels = \empty,
     		]

		\addplot [color = black, mark = square,      mark options = {mark size = 2, fill = black}, line width = \plotLineWidth] coordinates {
        (5, 3.2252)(6, 3.2869)(7, 3.3086)(8, 3.3173)(9, 3.3672)(10, 3.3468)(11, 3.3944)(12, 3.3576)(13, 3.3819)(14, 3.3449)(15, 3.3449)
		};		

		\addplot [color = blue,  mark = *, mark options = {mark size = 2, fill = blue},  line width = \plotLineWidth] coordinates {
		(5, 2.6643)(6, 2.5675)(7, 2.6229)(8, 2.4754)(9, 2.3749)(10, 2.3340)(11, 2.3008)(12, 2.1971)(13, 2.2525)(14, 2.1165)(15, 2.0733)
		};		

        \end{groupplot}
    \end{tikzpicture}
    
    \caption{\label{Fig:bpe}Normalized cost of the heterogeneous 3-caches baseline scenario for varying indicator sizes, measured by bits-per-cached-element (\bpe).
   }
\end{figure}
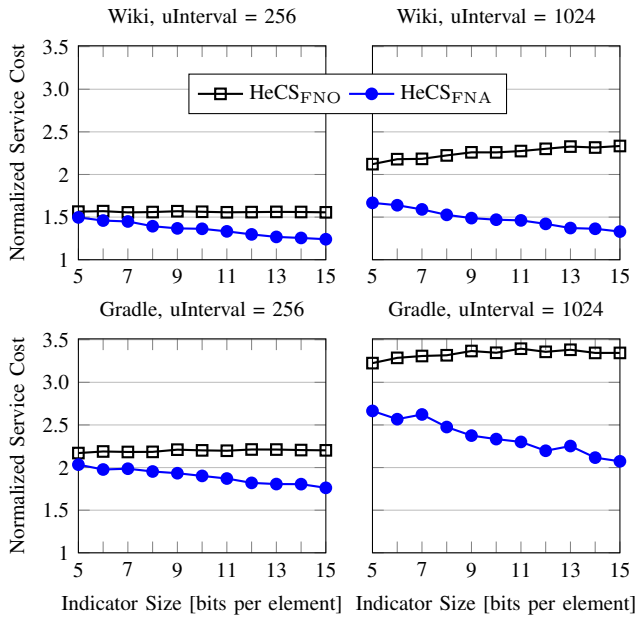
As expected, \pgmfna's performance improves when increasing the indicator size as larger indicators exhibit fewer false-positive errors. 
Interestingly, and somewhat counter-intuitively, there exist cases where the performance of \pgmfno\ does not improve when increasing the indicator size, and in some cases, performance actually {\em degrades}.
To explain this anomaly, let us understand the impact of false-positive and false-negative indications and their interplay.
First, note that the false-positive rate is often inversely proportional to the false-negative rate.
I.e., a constant decrease in the false-positive ratio is usually associated with an increase in the false-negative ratio. An extreme case occurs when all indications are negative, thus exhibiting a false-positive ratio of 0 and a sizeable false-negative ratio.
Next, note that a false-positive event typically translates to an unnecessary cache access, resulting in a relatively small penalty (e.g., an access cost of 1, 2, or 3 in our evaluation).
However, a false-negative event typically translates to a ``non-compulsory'' miss, translating to a high miss penalty (e.g., 100, in our evaluation).
It follows that even a mild decrease in the false-positive ratio may result in a non-negligible increase in the false-negative ratio that may nullify its benefits.  Such effects are especially significant when the miss penalty is high, which is common as misses often result in accessing memories whose access time may be orders of magnitude higher than that of the cache~\cite{adaptive_cache, justify_high_missp}.
Still, our proposed false-negative aware algorithm \pgmfna\ handles such scenarios seamlessly and benefits from the reduced false-positive ratio without adverse performance impact.




\subsection{Impact of Caching Capacity}
We now study the effect of having a larger or more diverse caching capacity on system performance.
For such an evaluation, we use 4.3M requests from the Wiki trace (instead of 1M).
Further, we now consider the {\em actual} mean cost per request (and not the normalized service cost) as the cost of \pif\  decreases when increasing the caching capacity. 


\subsubsection{Scaling the cache size}
\addtolength{\textfloatsep}{-0.2in}
\begin{figure}
    \begin{tikzpicture}
        \begin{groupplot}[
            group style=
                {
                columns = 2,
                xlabels at=edge bottom,
                ylabels at=edge left,
                horizontal sep=0.07\columnwidth,
                group name=plots
                },
    		width  = \plotWidth,
     		height = \plotHeight,
     		minor tick style={draw=none},
            ymajorgrids = true,
    		ylabel = Service Cost,
    		xmode = log,
    		ymin = 10,
    		ymax = 50,
    		ytick = {20, 30, 40},
    		yticklabels = {20, 30, 40},
    		xtick = {1, 2, 4, 8, 16, 32},
    		xticklabels = {1K, 2K, 4K, 8K, 16K, 32K},
            xmin = 1,
            xmax = 32,
     		ylabel near ticks,
            legend cell align={left},
    		legend style = {at={(-0.1, \plotLegendYOffset)}, anchor=north, legend columns=-1, font=\footnotesize},
    		xlabel = Cache Size,
    		xlabel near ticks,
    	    label style={font=\footnotesize},
    	    tick label style={font=\footnotesize},
    	]
        	
        \nextgroupplot[
            title = {update Interval = 256},
            title style ={
            font=\footnotesize,
            yshift = -2pt,
                },
     		]

		\addplot [color = green, mark=+, line width = \plotLineWidth] coordinates {
		(1, 34.8666)(2, 25.0614)(4, 15.2539)(8, 14.3832)(16, 13.4683)(32, 12.5000)(32, 12.5000)
		};		

		\addplot [color = black, mark = square,      mark options = {mark size = 2, fill = black}, line width = \plotLineWidth] coordinates {
		(1, 38.5447)(2, 27.9075)(4, 17.7871)(8, 16.5722)(16, 15.5163)(32, 15.1277)
		};		

		\addplot [color = blue,  mark = *, mark options = {mark size = 2, fill = blue},  line width = \plotLineWidth] coordinates {
		(1, 37.7650)(2, 27.4837)(4, 17.9997)(8, 15.9928)(16, 14.9447)(32, 13.8090)
		};		

        \nextgroupplot[
            title = {update Interval = 1024},
            title style ={
            font=\footnotesize,
            yshift = -2pt,
                },
     		]

		\addplot [color = green, mark=+, line width = \plotLineWidth] coordinates {
		(1, 34.8666)(2, 25.0614)(4, 15.2539)(8, 14.3832)(16, 13.4683)(32, 12.5000)(32, 12.5000)
		};		
		\addlegendentry {\pif}

		\addplot [color = black, mark = square,      mark options = {mark size = 2, fill = black}, line width = \plotLineWidth] coordinates {
		(1, 45.7890)(2, 35.3550)(4, 25.8854)(8, 23.7646)(16, 22.4458)(32, 21.0875)
		};		
		\addlegendentry {\pgmfno}

		\addplot [color = blue,  mark = *, mark options = {mark size = 2, fill = blue},  line width = \plotLineWidth] coordinates {
		(1, 39.4406)(2, 29.1637)(4, 18.5417)(8, 17.7776)(16, 17.1452)(32, 16.1261)
		};		
		\addlegendentry {\pgmfna}

        \end{groupplot}
    \end{tikzpicture}
    \caption{\label{Fig:cache_size}Cost of the heterogeneous 3-caches baseline for varying cache size.}
\end{figure}
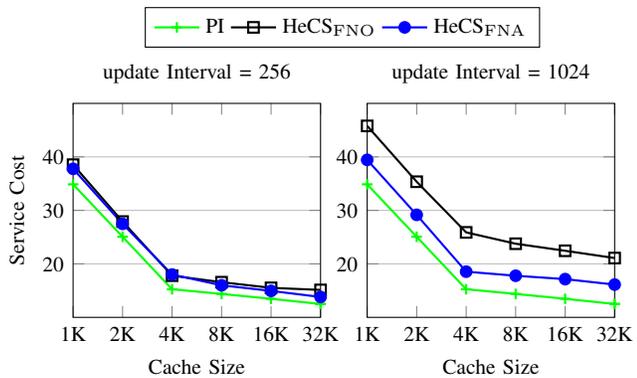
\begin{figure}
    \begin{tikzpicture}
        \begin{groupplot}[
            group style=
                {
                columns = 2,
                xlabels at=edge bottom,
                ylabels at=edge left,
                horizontal sep=0.07\columnwidth,
                group name=plots
                },
    		width  = \plotWidth,
     		height = \plotHeight,
     		minor tick style={draw=none},
            ymajorgrids = true,
    		ylabel = Service Cost,
            xmin = 1,
            xmax = 8,
    		xtick = {1, 2, 3, 4, 5, 6, 7, 8},
    		xticklabels = {1, 2, 3, 4, 5, 6, 7, 8},
     		ymin = 12,
     		ymax = 28,
     		ylabel near ticks,
            legend cell align={left},
    		legend style = {at={(-0.1, \plotLegendYOffset)}, anchor=north, legend columns=-1, font=\footnotesize},
    		xlabel = Number of Caches,
    		xlabel near ticks,
    	    label style={font=\footnotesize},
    	    tick label style={font=\footnotesize},
    	]
        	
        \nextgroupplot[
            title = {update Interval = 256},
            title style ={
            font=\footnotesize,
            yshift = -2pt,
                },
     		]

		\addplot [color = green, mark=+, line width = \plotLineWidth] coordinates {
		(1, 15.4953)(2, 14.6255)(3, 14.1027)(4, 13.7243)(5, 13.4282)(6, 13.1790)(7, 12.9644)(8, 12.7762)
		};

		\addplot [color = black, mark = square,      mark options = {mark size = 2, fill = black}, line width = \plotLineWidth] coordinates {
		(1, 16.3143)(2, 16.0114)(3, 16.2093)(4, 16.4719)(5, 16.9283)(6, 17.7537)(7, 18.3746)(8, 19.0733)
		};

		\addplot [color = blue,  mark = *, mark options = {mark size = 2, fill = blue},  line width = \plotLineWidth] coordinates {
		(1, 15.7881)(2, 15.3972)(3, 16.2298)(4, 16.0001)(5, 16.5326)(6, 17.3598)(7, 18.2090)(8, 18.9169)
		};		

        \nextgroupplot[
            title = {update Interval = 1024},
            title style ={
            font=\footnotesize,
            yshift = -2pt,
                 },
      		yticklabels = \empty,
     		]

		\addplot [color = green, mark=+, line width = \plotLineWidth] coordinates {
		(1, 15.4953)(2, 14.6255)(3, 14.1027)(4, 13.7243)(5, 13.4282)(6, 13.1790)(7, 12.9644)(8, 12.7762)
		};		
		\addlegendentry {\pif}

		\addplot [color = black, mark = square,      mark options = {mark size = 2, fill = black}, line width = \plotLineWidth] coordinates {
		(1, 18.0085)(2, 20.7707)(3, 23.2892)(4, 24.8814)(5, 25.7820)(6, 26.6699)(7, 27.1717)(8, 27.7612)
		};		
		\addlegendentry {\pgmfno}

		\addplot [color = blue,  mark = *, mark options = {mark size = 2, fill = blue},  line width = \plotLineWidth] coordinates {
		(1, 16.0475)(2, 15.8436)(3, 16.0115)(4, 17.7375)(5, 20.7803)(6, 23.1451)(7, 24.7284)(8, 25.8567)
		};		
		\addlegendentry {\pgmfna}

        \end{groupplot}
    \end{tikzpicture}
    \caption{\label{Fig:num_caches}Effect of varying the number of caches on the service cost. The access cost to each cache is 2, and the miss penalty is 100.}
\end{figure}
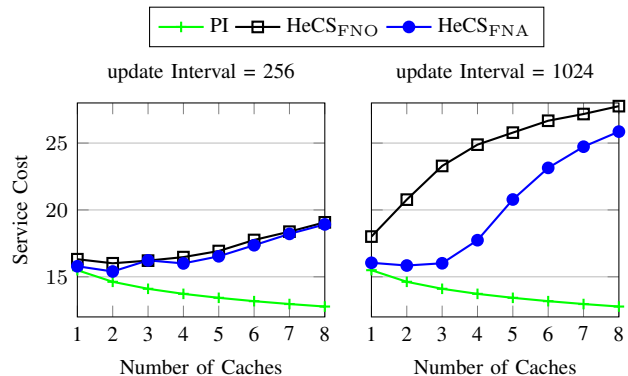

We study the impact of the cache size on the performance of \pgmfno\ and \pgmfna\ with update interval of 256 and 1024. The results in Fig.~\ref{Fig:cache_size} show that, as could be expected, for every given setting, scaling-up the caches' capacities decreases the service cost due to the improved hit ratio.
Our results show that when updates are relatively frequent (e.g., the case where the update interval is 256), the performance of \pgmfno\ is comparable to that of \pgmfna\, and they both exhibit a performance close to that of the ideal \pif\ strategy.
However, once the updates are less frequent (e.g., the case where the update interval is 1024), \pgmfno\ exhibits a significant degradation in performance.
\pgmfna, on the other hand, is far less affected by the increase in the update interval, and is still quite comparable to \pif.
In general, \pgmfna\ shows up to 25\% reduction in cost compared to \pgmfno.
The differences between \pgmfno\ and \pgmfna\ become more accentuated when one considers the cache size required to maintain a certain level of cost; \pgmfna\ performs better with 4K items caches, than \pgmfno\ with caches of size 32K. 


\subsubsection{Scaling the number of caches}
We now vary the number of caches in homogeneous settings.
All caches have an access cost of 2, ensuring that the average access cost is the same as in other scenarios examined in our evaluation.
Fig.~\ref{Fig:num_caches} shows the results for update intervals of 256 and 1024.
Notice that \pgmfna\ consistently outperforms \pgmfno, and the difference is more significant for large update intervals.
The results also imply that having more caches may hinder the performance of \pgmfna\ and \pgmfno.
Intuitively, in such a case, there are more false positives, and it is harder to guarantee that we access true positive items. Similarly, there are more negative indications which makes it harder for \pgmfna\ to identify a false-negative.

\revision{
\subsection{Impact of Exclusion Probabilities Estimations}
\label{sec:sim:estimations}

\begin{table}[t]
    \centering
    \revision{
    \caption{\label{tab:cache_aware}Normalized service cost of \pgmfno, \pgmfna, and \pgmfnastar\ for the baseline scneario.}
    \begin{tabular}{|c||c|c|c|}
    \hline 
    Trace &		 \pgmfno &	 \pgmfna\ 	 & \pgmfnastar \tabularnewline  \hline \hline 
    wiki	&	 2.2975 &	 1.3606 &	 1.1533   \tabularnewline \hline
    gradle	&	 3.3669 &	 2.1470 &	 1.3005   \tabularnewline \hline
    scarab	&	 2.9899 &	 1.6959 &	 1.0926   \tabularnewline \hline
    F2		&	 2.7732 &	 1.6695 &	 1.1079   \tabularnewline \hline\end{tabular}
    }
\end{table}

Our \pgmfna\ algorithm estimates the exclusion probabilities $\mrj$, as detailed in  Sec.~\ref{sec:hetero:estimation}. 
In this section, we study the effect of our estimations of the exclusion probabilities on the service cost. Specifically, we study a much stronger estimation (which is mostly impractical), and show that our approach captures the trends offered by such unrealistic estimations.

Upon every request, we define the {\em cache-aware} estimations of the exclusion probabilities.
In such an estimation of, say, $\mronej$, we consider the ratio between the number of false-positive indications of indicator $j$ and the total number of positive indications produced by indicator $j$. This ratio is computed for requests that arrived since the last advertisement was received. We estimate $\mrzeroj$ in a similar way.
We note that in order to obtain such estimations, one needs to know the {\em actual} content of each cache prior to every request (in order to distinguish between false and true indications), which is effectively impractical without accessing each cache for every request.

In what follows, we let \pgmfnastar\ denote our \pgmfna\ algorithm which uses the unrealistic cache-aware estimations, instead of the estimations described in Sec.~\ref{sec:hetero:estimation}.
Table~\ref{tab:cache_aware} shows the normalized service cost of \pgmfnastar, compared to that of \pgmfna\ and \pgmfno\ (for our baseline scenario).
Indeed, there is no wonder that the cache-aware estimations induce a lower cost than our proposed estimation method (as they are effectively all-knowledgeable in their estimations).
However, we can see that our approach indeed provides most of the benefits possible within such a system, compared to the FN-oblivious approach, while using a light-weight estimation method, that doesn't require or resort to having global and up-to-date information of the content of the caches.
}

\section{Conclusions}
This work studies the cache selection problem while using approximate indicators exhibiting both false-positive and false-negative errors.
The client in such a system selects a subset of the caches to minimize the expected service cost.
While there is extensive work in this field, all previous access strategies do not access caches with negative indications. 
While reasonable at first glance, our work shows that such an omission severely hinders the system's performance.
We argue that caches that periodically advertise their content indicators inherently introduce false-negative indications, and the rate of such indications is non-negligible.
In particular, we show that it is sometimes advisable to access caches with a negative indication, as it may reduce the overall system cost. 

We devise false-negative-aware access strategies in two main scenarios:
\begin{inparaenum}[(i)]
\item fully-homogeneous settings, where we show a policy that attains the optimal (minimal) access cost, and
\item general heterogeneous environments, where we present a strategy for which we can bound its approximation guarantee compared to the optimal solution.
\end{inparaenum}
We complete our study through an extensive evaluation based on real system traces.
Our results show that our proposed methods perform significantly better than the state-of-the-art in diverse settings.
Furthermore, our false-negative aware solutions can match the cost of competitive false-negative oblivious approaches while requiring an order of magnitude fewer resources (e.g., caching capacity or bandwidth required for indicators advertisement).

Our results demonstrate the potential benefits of embracing false-negative awareness into the algorithmic design space.
We expect our work to further induce both analytical and experimental research on the role of false-negatives in large distributed systems,
including dealing with non-homogeneous object size, adhering to bandwidth constraints, and studying correlated distributed caching schemes.

\bibliographystyle{IEEEtran}
\bibliography{Refs}

\end{document}